\documentclass{lmcs} 
\pdfoutput=1

\usepackage[utf8]{inputenc}

\usepackage{lastpage}
\lmcsdoi{17}{4}{2}
\lmcsheading{}{\pageref{LastPage}}{}{}%
{Jun.~12,~2020}{Oct.~01,~2021}{}


\usepackage{xcolor} 
\usepackage{mathtools}
\usepackage{amssymb}
\usepackage{amsmath}
\usepackage{stmaryrd}
\usepackage{xspace}
\usepackage[all,2cell]{xy}

\usepackage{pkg/math-gen}
\usepackage{pkg/is}
\usepackage{pkg/poset}

\newif\ifsubmit
\submittrue

\ifsubmit
\newcommand{\EZ}[1]{{#1}} 

\newcommand{\EZComm}[1]{} 
\newcommand{\FDComm}[1]{}
\else
\newcommand{\EZ}[1]{\textcolor{blue}{#1}} 
 
\newcommand{\EZComm}[1]{{\scriptsize\textcolor{blue}{[\bf{Elena: }#1}]}}
\newcommand{\FDComm}[1]{{\scriptsize\textcolor{red}{[\bf{Francesco: }#1}]}}
\fi

\newcommand{\Gr}{\mathit{G}} 
\newcommand{\node}{\mathit{v}} 
\newcommand{\anode}{\mathit{u}} 
\newcommand{\Edges}{\mathit{E}} 
\newcommand{\pthLen}[1]{\|{#1}\|} 
\newcommand{\gdist}[4]{\mathsf{dist}_{#1}(#2,#3,#4)} 
\newcommand{\Nodes}{\mathit{V}} 
\newcommand{\distis}{\is^{\mathsf{dist}}} 

\newcommand{\distSpec}{\mathcal{D}}
\newcommand{\distStar}{\ensuremath{(\star)}}

\newcommand{\addNum}[4]{\mathsf{add}(#1,#2,#3,#4)} 
\newcommand{\minElem}[2]{\mathsf{min}(#1,#2)} 
\newcommand{\allPos}[1]{\mathsf{allPos}(#1)} 

\newcommand{\posis}{\is^{>0}}

\newcommand{\posSpec}{\Spec^{>0}} 

\newcommand{\minSpec}{\Spec^{\mathsf{min}}} 
\newcommand{\addSpec}{\Spec^{\mathsf{add}}} 
\newcommand{\rep}[1]{\overline{#1}}

\newcommand{\toTreen}[1]{\mathsf{tr}^{#1}_{\Pair{X}{g}}} 
\newcommand{\toTree}[4]{\toTreen{#1}(#3,#4)} 

\newcommand{\rterm}[1]{\textsf{reg-red}(#1)}

\begin{document}

\title{Foundations of regular coinduction} 


\author[F.~Dagnino]{Francesco Dagnino}	
\address{DIBRIS, University of Genova, Italy}	
\email{francesco.dagnino@dibris.unige.it}  





\begin{abstract}
Inference systems are a widespread framework used to define possibly recursive predicates by means of inference rules. 
They allow both inductive and coinductive interpretations that are fairly well-studied. 
In this paper, we consider a middle way interpretation, called \emph{regular}, which combines advantages of both approaches: it allows non-well-founded reasoning while being finite. 
We show that the natural proof-theoretic definition of the regular interpretation, based on regular trees, coincides with a rational fixed point.  
Then, we provide an equivalent inductive characterization, which leads to an algorithm which looks for a regular derivation of a judgment.
Relying on these results, we define proof techniques for regular reasoning: the \emph{regular coinduction principle}, to prove completeness, and an inductive technique to prove soundness, based on the inductive characterization of the regular interpretation. 
Finally, we show the regular approach can be smoothly extended to \emph{inference systems with corules}, a recently introduced, generalised framework, which allows one to refine the coinductive interpretation, 
proving that also this flexible regular interpretation admits an equivalent inductive characterisation. 
\EZComm{sarebbe bello trovare un nome per la caratterizzazione  induttiva e la relativa tecnica di prova} 
\end{abstract}

\maketitle

%
%

\section{Introduction} \label{sect:intro}

Non-well-founded structures, such as graphs and streams,  are ubiquitous in computer science. 
Defining and proving properties of this kind of structures in a natural way is a challenging problem. 
Indeed, standard inductive techniques are not \EZ{adequate}, because they \EZ{require} to reach a base case in finitely many steps, and this is clearly not guaranteed\EZ{, since non-well-founded structures are conceptually infinite.}
The natural way to deal with such structures is by \emph{coinduction}, the dual of induction, which allows non-well-founded reasoning. 

A widespread approach to structure formal reasoning is by \emph{inference rules}, which define the steps we can do to prove judgements we are interested in. 
They support both inductive and coinductive reasoning in a pretty natural way: 
in inductive reasoning we are only allowed to use finite derivations, while in the coinductive one we can prove judgements by arbitrary, finite or infinite, derivations, 
 hence \EZ{coinductive reasoning} can properly handle non-well-founded structures. 

Coinductive reasoning is very powerful: it allows to derive judgements which require the proof of infinitely many different judgements. 
For instance, consider the following inference rule used to prove that a stream contains only positive elements: 
\[ \Rule{ \allPos{s} }{ \allPos{x\colon s} }\,x>0 \]
\EZ{To} prove that the stream of all odd natural numbers contains only positive elements, 
we can use the following  infinite derivation: 
\[ 
\Rule{
  \Rule{
    \Rule{\vdots}{ \allPos{5:7:9:\ldots} } 
  }{ \allPos{3:5:7:\ldots} } 
}{ \allPos{1:3:5\ldots} }  
\]
which is valid in coinductive reasoning and contains infinitely many different judgements. 

However, there are cases where, even though we need an infinite derivation, this derivation requires only the proof of \emph{finitely many} different judgements. 
This is often the case when dealing with cyclic structures, such as graphs or cyclic streams, since they are non-well-founded, but finitely representable.
For instance, if we want to prove that the stream of all 1's contains only positive elements, we \EZ{can use} the following derivation: 
\[
\Rule{
  \Rule{
    \Rule{\vdots}{ \allPos{1:1:1:\ldots} }
  }{ \allPos{1:1:1:\ldots} } 
}{ \allPos{1:1:1:\ldots} } 
\]
which is infinite, but requires only the proof of $\allPos{1:1:1:\ldots}$. 

Borrowing the terminology from trees \cite{Courcelle83}, we call a derivation requiring the proof of finitely many different judgments \emph{regular} (a.k.a. \emph{rational}\footnote{The terms regular and rational are synonyms. However we will mainly use the second one for the model-theoretic approach, see \refToSect{rfp}.}), 
and we call \emph{regular coinduction} (or \emph{regular reasoning})  the approach that allows only regular derivations. 

Whereas inductive and coinductive reasoning have well-known  semantic foundations and proof principles, to our knowledge regular reasoning by means of  inference rules  has never been explored at the same extent. \EZ{The aim of this paper is to fill this gap, by providing solid foundations also to the regular approach. Indeed, we believe that}
the regular approach provides a very interesting middle way between induction and coinduction. 

Indeed, inductive reasoning is  restricted to finite derivations, but, in return, we implicitly get an (abstract) algorithm, which looks for a derivation of a judgement. 
Such an algorithm is sound and complete with respect to derivable judgements. \EZ{That is, it may not terminate for judgements that do not have a finite derivation, but it is guaranteed to successfully terminate, finding a finite derivation, for all and only derivable judgments.}
Instead, coinductive reasoning allows also infinite derivations, but there is no hope, in general, to find an algorithm \EZ{which succesfully terminates for derivable judgments,} because, as we have seen,  a derivation may require infinitely many different judgements to be proved\footnote{This is just an intuitive explaination. This fact has been proved for logic programs in \cite{AnconaD15}, which are a particular, syntactic, instance of general rule-based definitions considered in this paper. }. 

Regular reasoning combines advantages of the two approaches:
on one hand, it is not restricted to finite derivations, going beyond limits of induction, but, on the other hand, it still has, like induction,  a finite nature, 
\EZ{hence it is possible to design an algorithm which finds a derivation for all and only derivable judgments, as we will show in the following.}

In detail, the contribution of this paper is the following.
\begin{itemize}
\item An equivalent \emph{model-theoretic} characterization of judgements derivable by a regular proof tree, showing it is an instance of the \emph{rational fixed point} \cite{AdamekMV06}.  
\EZ{This is important since it provides a purely semantic view of regular coinduction. Moreover, from}
this we get a proof principle, the \emph{regular coinduction principle}, which can be used to prove completeness of a set of inference rules against a set of valid judgements, that is, that all valid judgement are derivable by a regular proof tree. 
\item An equivalent \emph{inductive} characterization of judgements derivable by a regular proof tree. 
Essentially, following the structure of the operational model of coinductive logic programming \cite{SimonMBG06,AnconaD15}, 
but in the more abstract setting of rule-based definitions, 
we enrich judgements by a finite set of \emph{circular hypotheses}, used to keep  track of already encountered judgements so that, when the same judgement is found again, it can be used as an axiom. 
\EZ{This nicely formalizes, by an abstract construction in the general setting of rule-based definitions and a correcteness proof given once and for all, techniques used in different specific cases for dealing with cyclic structures inductively, by detecting cycles to ensure termination. }
Furthermore, this provides us with a sound and complete algorithm \EZ{to find a regular derivation for a judgment, if any.}
Finally, relying on the inductive characterization, we define a proof technique to show soundness of a set of inference rules against a set of valid judgements, that is, that  all derivable judgements are valid. 
\end{itemize}


Moreover, we show that all these results can be smoothly extended to a recently introduced generalisation of coinductive reasoning, \EZ{namely} \emph{flexible coinduction} \cite{Dagnino17,AnconaDZ17esop,Dagnino19}. 
Beside standard rules, this generalised framework allows  also \emph{corules}, 
which are special rules used to validate infinite derivations. 
As a result, using corules, we are able to filter out some undesired infinite derivations, having a much finer control on the set of derivable judgements. 
Flexible coinduction  smoothly extends standard coinduction, subsuming it, that is, \EZ{standard coinduction} can be recovered by a specific choice of corules. 
We will show this is the case also in the regular case, that is, flexible regular coinduction subsumes standard regular coinduction.

The rest of the paper is organized as follows. 
In \refToSect{is} we recall basic concepts on inference systems and define the regular interpretation in proof-theoretic terms. 
In \refToSect{rfp} we define the rational fixed point in a lattice-theoretic setting,  and in \refToSect{equiv} we prove that the regular interpretation coincides with a rational fixed point.
\refToSect{cycle} provides the equivalent inductive characterization of the regular interpretation and \refToSect{reasoning} discusses proof techniques for regular reasoning.
In \refToSect{corules} we extend all the previously presented results to flexible coinduction. \EZ{Finally, \refToSect{related}  discusses related work and 
\refToSect{conclu} concludes the paper, outlining future work. }

\subsubsection*{Notations}
Let $X$ be a set, we denote by $\wp(X)$ and $\finwp(X)$  its power-set and its finite power-set, respectively. 
For a function  $\fun{f}{X}{Y}$, $\fun{\img{f}}{\wp(X)}{\wp(Y)}$ and $\fun{\inv{f}}{\wp(Y)}{\wp(X)}$ are the direct image and the inverse image along $f$, respectively, hat is, 
$\img{f}(A) = \{ y \in Y \mid \exists x \in A. y = f(x) \}$ and 
$\inv{f}(B) = \{ x \in X \mid f(x) \in B \}$. 
We also denote by $\List{X}$ the set of finite sequences on $X$, by $\EList$ the empty sequence and by $\alpha\beta$ the concatenation of sequences $\alpha$ and $\beta$.

\section{Inference systems and regular derivations}  \label{sect:is} 

In this section, we recall basic definitions about inference systems \cite{Aczel77,LeroyG09,Sangiorgi11} and their standard semantics, 
and define their regular interpretation. 

Let us assume a \emph{universe} $\universe$, which is a set, whose elements $\judg$ are called \emph{judgements}. 
An \emph{inference system} $\is$ is a set of \emph{(inference) rules}, which are pairs  $\RulePair{\prem}{\conclu}$, also written $\Rule{\prem}{\conclu}$,  where $\prem \subseteq \universe$ is the set of \emph{premises}, while $\conclu \in \universe$ is the \emph{conclusion}. 
A \emph{proof tree} (a.k.a. \emph{derivation}) in $\is$  is a tree with nodes (labelled) in $\universe$ and such that, for each node $\judg$ with set of children $\prem$, there is a rule $\RulePair{\prem}{\judg}$ in $\is$. 
The \emph{inference operator} $\fun{\InfOp{\is}}{\wp(\universe)}{\wp(\universe)}$ is defined as follows: 
\[ \InfOp{\is}(X) = \{ \judg \in \universe \mid \exists \prem \subseteq X.\ \RulePair{\prem}{\judg} \in \is \} \] 
A subset $X\subseteq \universe$ is \emph{closed} if, for all rules $\RulePair{\prem}{\judg} \in \is$, if $\prem\subseteq X$ then, $\judg\in X$, that is,  $\InfOp{\is}(X) \subseteq X$, 
it is \emph{consistent} if, for all $\judg\in X$, there is a rule $\RulePair{\prem}{\judg} \in \is$, that is, $X\subseteq \InfOp{\is}(X)$, and 
it is an \emph{interpretation}, if it is both closed and consistent, namely, 
a fixed point $X = \InfOp{\is}(X)$. 

There are two main approaches to define interpretations of an inference system: the model-theoretic and the proof-theoretic one. 
The two standard interpretations, the inductive and the coinductive one, can be equivalently defined in proof-theoretic and in model-theoretic style: 
\begin{itemize}
\item the \emph{inductive interpretation} $\Ind{\is}$ is the set of judgements having a well-founded proof tree, \EZ{and} also the least fixed point of $\InfOp{\is}$, and 
\item the \emph{coinductive interpretation} $\CoInd{\is}$ is the set of judgements having an arbitrary (well-founded or not) proof tree, \EZ{and} also the greatest fixed point of $\InfOp{\is}$. 
\end{itemize}
In the following, we will write $\validInd{\is}{\judg}$ for $\judg \in \Ind{\is}$ and $\validCo{\is}{\judg}$ for $\judg \in \CoInd{\is}$. 

In this paper, we assume inference systems to be \emph{finitary}, that is, all rules have  a finite set of premises. 
Under this assumption, well-founded proof trees are always finite and infinite proof trees are always non-well-founded\footnote{This is an immediate consequence of the K\"{o}nig's lemma.}, hence we will use this simpler terminology.  

In the coinductive interpretation, since we allow arbitrary proof trees, we can derive judgements requiring infinitely many different judgements to be proved. 
However, there are cases where we still need infinite derivations, but only of finitely many judgements. 
This idea of an infinite proof tree containing only finitely many different judgements nicely corresponds to a well-known class of trees: \emph{regular trees} \cite{Courcelle83}. 
We say that a tree is \emph{regular}  if it has a finite number of different subtrees. 
Then, we can define another set of judgements: 

\begin{defi}[Regular interpretation] \label{def:rational-is}
The \emph{regular interpretation} of an inference system $\is$ is the set $\Reg{\is}$ of judgements having a regular proof tree.  
\end{defi}
In the following we will write $\validReg{\is}{\judg}$ for $\judg \in \Reg{\is}$. 
To \EZ{ensure} that the regular interpretation is well-defined, we have to check it is indeed an interpretation, namely, it is a fixed point of $\InfOp{\is}$. We refer to \refToSect{equiv} for this proof. 

Let us illustrate \EZ{regular proof trees by} a couple of examples. 
Consider $\lambda$-terms, ranged over by  $M,N$, with the usual full $\beta$-reduction $\to$. 
Denote by $R(M)$ the set of terms reachable from $M$, that is, $R(M) = \{ N \mid M\to^\star N\}$, where $\to^\star$ is the reflexive and transitive closure of $\to$, and by $S(M)$ the subset of $R(M)$ of those terms reachable in one step, that is, $S(M) = \{ N \mid M \to N\}$. 
Note that $S(M)$ is necessarily finite as the set of redexes in any $\lambda$-term is finite, while $R(M)$ may in general be infinite. 
We say that a term $M$ is \emph{regularly reducible} if the set $R(M)$ is finite. 
In other words, this means that its evaluation graph of $M$ is finite or, alternatively, its evaluation tree is regular, because the evaluation tree of $M$ has one different subtree for any term reachable from $M$. 
Obviously, all strongly normalising terms are regularly reducible, but also some non-normalising terms are regularly reducible, 
for instance, so is the term $\Omega = \Delta\, \Delta$, where $\Delta = \lambda x.x\,x$, as it always reduces to itself. 
Instead, a term like $\Omega_1 = \Delta_1\,\Delta_1$, with $\Delta_1 = \lambda x.x\,x\,x$, is not regularly reducible 
as it reduces to larger and larger terms 
($\Omega_1\to\Omega_1\, \Delta_1 \to \Omega_1\, \Delta_1\, \Delta_1 \to \ldots$). 

We can define a judgement $\rterm{M}$ characterising regularly reducible terms as the regular interpretation of the following rule: 
\[
\MetaRule{step}{
  \rterm{N_1}\Space\ldots\Space \rterm{N_k} 
}{ \rterm{M} }{ S(M) = \{ N_1,\ldots,N_k \} } 
\]
Indeed, the proof tree for $\rterm{M}$ coincides with the evaluation tree of $M$, hence judgements derivable by regular proof trees coincides with regularly reducible terms.

\label{page:ex-dist}
As another example, 
assume that we want to define the judgement $\gdist{\Gr}{\node}{\anode}{\delta}$, where $\Gr$ is a graph, $\node$ and $\anode$ are nodes in $\Gr$ and $\delta \in [0,\infty]$, stating that the distance from $\node$ to $\anode$ in $\Gr$ is $\delta$. 
We represent a graph by its accessibility function $\fun{\Gr}{\Nodes}{\wp(\Nodes)}$, where $\Nodes$ is the finite  set of nodes.
The judgement is defined by the following (meta-)rules, where we assume $\min \emptyset = \infty$:   
\[ 
\MetaRule{empty}{}{ \gdist{\Gr}{\node}{\node}{0} }{}  
\BigSpace 
\MetaRule{adj}{
  \gdist{\Gr}{\node_1}{\anode}{\delta_1} \Space \ldots \Space \gdist{\Gr}{\node_n}{\anode}{\delta_n}
}{ \gdist{\Gr}{\node}{\anode}{1 + \min \{\delta_1, \ldots, \delta_n\} }  }{ 
\node \ne \anode \\ 
\Gr(\node) = \{\node_1,\ldots,\node_n\} 
}
\]
Of course the inductive interpretation is not enough: it can only deal with acyclic graphs, because, in presence of cycles, we cannot reach a base case (an axiom) in finitely many steps. 
Hence, we need infinite derivations to handle cycles, and, 
since the set of nodes is finite, to compute the distance, we need only finitely many judgements, thus regular derivations should be enough. 
\refToFig{ex-dist} shows a concrete example of this feature of derivations in the above inference system. 

\begin{figure}
\[\begin{array}{cc}
\vcenter{\xymatrix@R=10ex@C=10ex{
a \ar@/^5pt/[r] \ar[d] & b \ar@/^5pt/[l] \ar[d] \\
d \ar@(dr,ur)   & c
} }  & 
\Rule{
  \Rule{
    \Rule{\vdots}{ \gdist{\Gr}{a}{c}{2} }
    \Space 
    \Rule{}{ \gdist{\Gr}{c}{c}{0} } 
  }{ \gdist{\Gr}{b}{c}{1} }  \Space 
  \Rule{
    \Rule{\vdots}{ \gdist{\Gr}{d}{c}{\infty} }
  }{ \gdist{\Gr}{d}{c}{\infty} }
}{ \gdist{\Gr}{a}{c}{2}  }
\end{array} \]
\caption{On the left side a concrete graph $\Gr$ with nodes $\{a,b,c,d\}$, and on the right side the regular derivation of the judgement $\gdist{\Gr}{a}{c}{2}$. } \label{fig:ex-dist} 
\end{figure}

As we said, standard inductive and coinductive interpretations are fixed points of the inference operator. 
In the next few sections, we will show that this is the case also for the regular interpretation.

\section{The rational fixed point} \label{sect:rfp} 

In this section we define the rational fixed point in a lattice-theoretic setting, which will be the basis for the fixed point characterisation of \EZ{the regular interpretation}. 
The construction we present in \refToDef{rfp} and \refToThm{rfp} is an instance of analogous constructions \cite{AdamekMV06,MiliusPW16,MiliusPW19} developed in a more general category-theoretic setting. 
We work in the lattice-theoretic setting, since it is enough for the aim of this paper and definitions and proofs are simpler and understandable by a wider audience. 

First, we report some basic definitions on lattices, for details we refer to \cite{DaveyP02}. 
A \emph{complete lattice} is a partially ordered set $\Pair{\lattice}{\order}$ where all $A \subseteq \lattice$ have a least upper bound (a.k.a.\ \emph{join}), denoted  by $\lub A$.
In particular, in $\lattice$  there are both a top element $\top = \lub \lattice$ and a bottom element $\bot = \lub \emptyset$.
Furthermore, it can be proved that all $A \subseteq \lattice$ have also a greatest lower bound (a.k.a.\ \emph{meet}), denoted by $\glb A$.
In the following, for all $x, y \in \lattice$,  we will write $x \join y$ for the binary join and $x \meet y$ for the binary meet. 
The paradigmatic example of complete lattice is the power-set $\wp(X)$  of a set $X$, ordered by set inclusion, where least upper bounds are given by unions. 

An element $x \in \lattice$ is \emph{compact} if, for all $A \subseteq \lattice$ such that $x \order \lub A$, there is a finite subset $B \subseteq A$ such that $x \order \lub B$. 
We denote by $\Compact{\lattice}$ the set of compact elements in $\lattice$. 
It is easy to check that  $\Compact{\lattice}$ is closed under binary joins, that is, if $x, y \in \lattice$ are compact, then  $x \join y$ is compact as well. 
In the power-set lattice, compact elements are finite subsets. 

An \emph{algebraic lattice} is a complete lattice $\Pair{\lattice}{\order}$ where each $x \in \lattice$ is the join of all compact elements below it, that is, 
$x = \lub \{ y \in \Compact{\lattice} \mid y \order x \}$. 
In other words, an algebraic lattice is generated by the set of its compact elements, since each element can be decomposed as a (possibly infinite) join of compact elements. 
The power-set lattice is algebraic, since each element can be decomposed as a union of singletons, which are obviously compact. 

Given a function $\fun{\MFun}{\lattice}{\lattice}$ and an element $x \in \lattice$, we say that
$x$ is a \emph{pre-fixed point} if $\MFun(x) \order x$, 
a \emph{post-fixed point} if $x \order \MFun(x)$, and 
a \emph{fixed point} if $x = \MFun(x)$.
We are interested in a special class of functions, called finitary functions (a.k.a. Scott-continuous functions), defined below. 
A subset $A \subseteq \lattice$ is \emph{directed} if, for all $x, y \in A$, there is $z \in A$ such that $x \order z$ and $y \order z$, then 
a \emph{finitary} function $\fun{\MFun}{\lattice}{\lattice}$ is a function preserving the joins of all directed subsets of $\lattice$, that is, 
for each directed subset $A \subseteq \lattice$, $\MFun(\lub A) = \lub \img{\MFun}(A)$. 
A finitary function is also monotone: if $x \order y$ then the set $\{x, y\}$ is directed and its join is $y$, hence we get $\MFun(y) = \MFun(x) \join \MFun(y)$, that is, $\MFun(x) \order \MFun(y)$.   
Monotone functions over a complete lattice have an important property:  thanks to the Knaster-Tarski theorem~\cite{Tarski55}, we know that they have both least and greatest fixed points, {that} we denote by $\lfp\MFun$ and $\gfp\MFun${,} respectively.
We will show that for a finitary function over an algebraic lattice we can construct another fixed point lying between the least and the greatest one. 
In the following we assume an algebraic lattice $\Pair{\lattice}{\order}$. 

\begin{defi} \label{def:rfp}
Let $\fun{\MFun}{\lattice}{\lattice}$ be a finitary function. 
The \emph{rational fixed point} of $\MFun$, denoted by $\rfp \MFun$, is the join of all compact post-fixed points of $\MFun$, that is, 
if $R_\MFun = \{ x \in \Compact{\lattice} \mid x \order \MFun(x) \}$, 
\[ \rfp \MFun = \lub R_\MFun  \]
\end{defi}
Note that, since both compact elements and post-fixed points are closed under binary joins, we have that, for all $x, y \in R_\MFun$, $x\join y \in R_\MFun$, but, in general, $\rfp \MFun$ is not compact, because it is the join of an infinite set. 

The following theorem ensures that the rational fixed point is well-defined, that is, it is indeed a fixed point. 
Such result is a consequence of a general category-theoretic analysis \cite{AdamekMV06}, we rephrase the proof in our more specific setting as it is much simpler. 

\begin{thm} \label{thm:rfp}
Let $\fun{\MFun}{\lattice}{\lattice}$ be a finitary function, then $\rfp \MFun$ is a fixed point of $\MFun$. 
\end{thm}
\begin{proof}
Since $\rfp\MFun$ is defined as the least upper bound of a set of post-fixed points, it is post-fixed as well. 
Hence, we have only to check that $\MFun(\rfp\MFun) \order \rfp\MFun$. 

First, since $\lattice$ is algebraic we have $\MFun(\rfp\MFun) = \lub \{ x \in \Compact{\lattice} \mid x \order \MFun(\rfp\MFun)\}$, hence it is enough to prove that, for all $x \in \Compact{\lattice}$ such that $x \order \MFun(\rfp\MFun)$, we have $x \order \rfp\MFun$. 
Consider $x \in \Compact{\lattice}$ such that $x \order \MFun(\rfp\MFun)$. 
Note that $R_\MFun$ is a directed set, indeed, if $X \subseteq  R_\MFun$ is finite, then $\lub X \in R_\MFun$,  hence, since $\MFun$ is finitary, we have $\MFun(\rfp\MFun) = \MFun(\lub R_\MFun) = \lub \img{\MFun}(R_\MFun)$. 
Therefore, $x \order \lub \img{\MFun}(R_\MFun)$ and, since $x$ is compact, there is a finite subset $W \subseteq R_\MFun$ such that $x \order \lub \img{\MFun}(W)$. 
Set $w = \lub W$, since $\MFun$ is monotone, we get $x \order \lub \img{\MFun}(W) \order \MFun(\lub W) = \MFun(w)$. 
By definition, $w\in R_\MFun$, namely, it is compact and post-fixed, hence we get $x\join w \order \MFun(w) \order \MFun(x\join w)$, since $\MFun$ is monotone. 
Finally, $x\join w $ is compact, as it is the join of compact elements, hence $x \join w \in R_\MFun$, and this implies $x \order x \join w \order \rfp \MFun$, as needed. 
\end{proof}

As for least and greatest fixed points, as an immediate consequence of \refToDef{rfp}, we get a proof principle to show that an element is below $\rfp\MFun$.
\begin{prop} \label{prop:below-rfp}
Let $\fun{\MFun}{\lattice}{\lattice}$ be a finitary function and $z \in \lattice$, then, 
 if there is a set $X \subseteq \Compact{\lattice}$ such that 
\begin{itemize}
\item for all $x \in X$, $x \order \MFun(x)$, and 
\item $z \order \lub X$,
\end{itemize}
then $z \order \rfp \MFun$. 
\end{prop} 
\begin{proof}
If these conditions hold, then we have $X \subseteq R_\MFun$, hence $z\order \lub X \order \lub R_\MFun = \rfp \MFun$. 
\end{proof}

\section{Fixed point semantics for regular coinduction} \label{sect:equiv}

In this section, we prove that the regular interpretation $\Reg{\is}$ of a (finitary)  inference system $\is$ (\refToDef{rational-is}) coincides with the rational fixed point of the inference operator $\InfOp{\is}$. 
\EZ{Rather than giving an ad-hoc proof, we present a general framework where to express in a uniform and systematic way the equivalence between proof-theoretic and model-theoretic semantics, and then prove such equivalence for the regular case. 
To do this, we need a more formal account of proof trees.} 
These definitions work for any inference system, even non-finitary ones. 

\subsubsection*{A formal account of proof trees}
To carry out the proof, we need a more formal account of proof trees (see \cite{Dagnino19} for details). 
\EZComm{c'era una parte ripetuta due volte, ho recuperata la nota della prima versione e l'ho inserita sotto ma solo perch\'e non sapevo se la volevi, \`e uguale}

A \emph{tree language} on a set $A$ is a non-empty and prefix-closed subset $L \subseteq \List{A}$, that is,  such that, for all $\alpha \in \List{A}$ and $x \in A$, if $\alpha x \in L$ then $\alpha \in L$, hence, in particular, the empty sequence belongs to any tree language. 

A  \emph{tree} $\tr$ on a set $A$ is a pair $\Pair{r}{L}$ where $L$ is a tree language on $A$ and $r\in A$ is the \emph{root} of the tree. 
We set $\TrNodes{\tr} = L$ and $\rt(\tr) = r$. 
Intuitively, a sequence $\alpha \in L$ represents a node of the tree labelled by $\tr(\alpha)$, defined as the last element of the sequence $\rt(\tr)\alpha$. 
Therefore, a tree $\tr$ on $A$ induces a partial function from $\List{A}$ to $A$ whose domain is  a tree language. 
Differently from the literature \cite{Courcelle83,AczelAMV03}, our definition  forces trees to be unordered and, more importantly, it ensures there cannot be two sibling nodes with the same label, we refer to \cite{Dagnino19}\EZ{, where these trees are called children injective,} for a detailed comparison. 
These two additional requirements will be essential to prove a crucial 
result of this section, \refToThm{universal-tree}, and \EZ{are} reasonable to define proof trees as we will see below. 

Given a tree $\tr$ and a node $\alpha \in \TrNodes{\tr}$, we denote by $\subtr{\tr}{\alpha}$ the \emph{subtree of $\tr$ rooted at $\alpha$}, defined as the pair $\Pair{\tr(\alpha)}{\{\beta \in \List{A} \mid \alpha\beta \in L\}}$.
Hence $\tr$ is \emph{regular} iff the set $\SubTr{\tr} = \{ \subtr{\tr}{\alpha} \mid \alpha \in \TrNodes{\tr}  \}$ is finite. 
We also define $\chltr{\tr}{\alpha} =  \{ \subtr{\tr}{\beta}  \mid \exists x\in A. \beta = \alpha x,\beta\in\TrNodes{\tr} \}$ the set of \emph{children} of $\alpha$ in $\tr$ and 
$\dsubtr{\tr} = \chltr{\tr}{\EList}$ the set of \emph{direct subtrees} of $\tr$, which are the children of the root of $\tr$. 
Note that, for all $\alpha \in \TrNodes{\tr}$, we have $\tr(\alpha) = \rt(\subtr{\tr}{\alpha})$ and $\chltr{\tr}{\alpha} = \dsubtr{\subtr{\tr}{\alpha}}$. 

Having these notations, we say that a tree $\tr$ on the universe $\universe$ is a \emph{proof tree} in $\is$ iff, for all nodes $\alpha \in \TrNodes{\tr}$,  we have 
$\RulePair{\rtdir(\chltr{\tr}{\alpha})}{\tr(\alpha)} \in \is$. 
In the following, as it is customary,  we often represent proof trees using stacks of rules, that is, if $\RulePair{\prem}{\conclu} \in \is$ and $T = \{\tr_i\mid i \in I\}$ is a collection of trees such that $\rtdir(T) = \prem$ and $\rt(\tr_i) = \rt(\tr_j)$ implies $i=j$, we denote by 
$\Rule{T}{\conclu}$ the proof tree $\tr = \Pair{\conclu}{\TrNodes{\tr}}$ where  
\[
\TrNodes{\tr} = \{\EList\} \cup \bigcup_{i \in I} \rt(\tr_i) \TrNodes{\tr_i} 
\]
Finally, we say that a tree $\tr$ is a \emph{proof tree for a judgement} $\judg \in \universe$ if it is a proof tree rooted in $\judg$, that is, $\rt(\tr) = \judg$.

\subsubsection*{Relating proof-theoretic and model-theoretic semantics}
We now present a systematic approach to define and relate proof-theoretic and model-theoretic semantics of inference systems. 
Let $\ProofT{\is}$ be the set of all (well-founded or not) proof trees in $\is$ and, abusing a bit the notation, let $\fun{\rt}{\ProofT{\is}}{\universe}$ be  the function that maps a proof tree to its root. 
Then the direct image and the inverse image along $\rt$ are $\fun{\rtdir}{\wp(\ProofT{\is})}{\wp(\universe)}$ and $\fun{\rtinv}{\wp(\universe)}{\wp(\ProofT{\is})}$, respectively. 
There is an adjunction $\rtdir \dashv \rtinv$, that is, for all $X\subseteq \ProofT{\is}$ and $Y\subseteq \universe$, $\rtdir(X) \subseteq Y$ iff $X \subseteq \rtinv(Y)$. 
In other words, $\rtdir$ behaves as an abstraction function \cite{CousotC77}. 
We define the \emph{tree inference operator} $\fun{\TInfOp{\is}}{\wp(\ProofT{\is})}{\wp(\ProofT{\is})}$ as follows: 
\[ \TInfOp{\is}(X) = \{ \tr \in \ProofT{\is} \mid \dsubtr{\tr} \subseteq X \mbox{ and } \RulePair{\rtdir(\dsubtr{\tr})}{\rt(\tr)} \in \is \}  \] 
It is easy to check that $\TInfOp{\is}$ is monotone. 
This function behaves the same way as $\InfOp{\is}$, but keeps track of the trees used to derive the premises of the rule. 
More formally, $\InfOp{\is}$ and $\TInfOp{\is}$ are related by the following proposition, which is an immediate consequence of their definitions. 

\begin{prop} \label{prop:infop-tinfop}
$\rtdir \circ \TInfOp{\is} = \InfOp{\is} \circ \rtdir$ and 
$\TInfOp{\is} \circ \rtinv \subseteq \rtinv \circ \InfOp{\is}$. 
\end{prop}

\begin{cor} \label{cor:infop-tinfop}
Let $X \subseteq \ProofT{\is}$ and $Y\subseteq \universe$, then 
\begin{itemize}
\item if $X\subseteq \TInfOp{\is}(X)$ then $\rtdir(X) \subseteq \InfOp{\is}(\rtdir(X))$, 
\item if $\TInfOp{\is}(X)\subseteq X$ then $\InfOp{\is}(\rtdir(X))\subseteq \rtdir(X)$, and 
\item if $\InfOp{\is}(Y) \subseteq Y$ then $\TInfOp{\is}(\rtinv(Y)) \subseteq \rtinv(Y)$. 
\end{itemize}
\end{cor}

It can be proved\footnote{We omit the proofs because the focus of the paper is on regular trees.\EZComm{va chiarito; vuoi dire che ci sono da qualche parte, o meglio che prove classiche potrebbero facilmente essere rifrasate cos\`i?}}  that $\lfp\TInfOp{\is}$ is the set of well-founded proof trees, while $\gfp\TInfOp{\is}$ is the set of all proof trees. 
Then, the correspondence between proof-theoretic and model-theoretic approaches in the standard inductive and coinductive cases can be \EZ{succinctly expressed} by the equalities: 
$\Ind{\is} = \rtdir(\lfp\TInfOp{\is}) = \lfp\InfOp{\is}$ and $\CoInd{\is} = \rtdir(\gfp\TInfOp{\is}) = \gfp\InfOp{\is}$.

\subsubsection*{The regular case}
First of all, we state a fundamental property of finitary inference systems:

\begin{prop} \label{prop:finitary-is} 
If $\is$ is finitary, then $\InfOp{\is}$ and $\TInfOp{\is}$ are finitary. 
\end{prop}
\begin{proof}
We do the proof for $\TInfOp{\is}$. 
Let $X\subseteq \wp(\universe)$ be a directed subset. 
Since $\TInfOp{\is}$ is monotone, we have $\bigcup \img{{\TInfOp{\is}}}(X) \subseteq \TInfOp{\is}(\bigcup X)$, hence we have only to check the other inclusion. 
If $\tr \in \TInfOp{\is}(\bigcup X)$,  $\dsubtr{\tr}\subseteq \bigcup X$ and, since $\dsubtr{\tr}$ is finite, as $\tr$ is a proof tree and $\is$ is finitary, there is a finite subset $Y\subseteq X$ such that $\dsubtr{\tr} \subseteq \bigcup Y$. 
Then, since $X$ is directed, there is $A \in X$ such that $\bigcup Y \subseteq A$, thus $\dsubtr{\tr} \subseteq A$, and so $\tr \in \TInfOp{\is}(A) \subseteq \bigcup \img{{\TInfOp{\is}}}(X)$, as needed. 
\end{proof}

Let us denote by $\RegT{\is}$ the set of regular proof trees in $\is$. 
Thanks to \refToProp{finitary-is} and \refToThm{rfp}, we know that $\rfp\InfOp{\is}$ and $\rfp \TInfOp{\is}$ are both well-defined. 
Then, the theorem we have to prove is the following:
 
\begin{thm} \label{thm:rational-sem}
$\Reg{\is} = \rtdir(\rfp\TInfOp{\is}) = \rfp\InfOp{\is}$. 
\end{thm}

The first step of the proof is to show that the set of regular proof trees coincides with the rational fixed point of $\TInfOp{\is}$. 
To this end, we have the following characterization of post-fixed points of $\TInfOp{\is}$. 

\begin{lem} \label{lem:tinfop-post-fixed}
Let $X \subseteq \ProofT{\is}$ be a set of proof trees, then  $X \subseteq \TInfOp{\is}(X)$ if and only if  
\[ X = \bigcup_{\tr \in X} \SubTr{\tr} \]
\end{lem}
\begin{proof}
We start from the left-to-right implication. 
The inclusion $\subseteq$ is trivial, since $\tr \in \SubTr{\tr}$ for any tree $\tr$. 
To prove the other inclusion, \EZ{set} $\tr \in X$, we have to show that, for all $\alpha \in \TrNodes{\tr}$, $\subtr{\tr}{\alpha} \in X$. 
The proof is by induction on $\alpha$. 
\begin{description}
\item [Base] If $\alpha$ is empty, then $\subtr{\tr}{\alpha} = \tr \in X$ by hypothesis. 
\item [Induction] If $\alpha = \beta \judg$, then $\subtr{\tr}{\beta\judg} = \subtr{(\subtr{\tr}{\alpha})}{\judg}$ and $\subtr{\tr}{\alpha} \in X$ by induction hypothesis. 
Since $X \subseteq \TInfOp{\is}(X)$, we have  $\subtr{(\subtr{\tr}{\alpha})}{\judg} \in\dsubtr{\subtr{\tr}{\alpha}} \subseteq X$, as needed. 
\end{description}
The other implication is trivial: if $\tr \in X$, $\dsubtr{\tr}\subseteq X$ by hypothesis and $\Rule{\rtdir(\dsubtr{\tr})}{\rt(\tr)} \in \is$, as $\tr$ is a proof tree, hence $\tr \in \TInfOp{\is}(X)$.
\end{proof}

\begin{lem} \label{lem:regtr-rfp}
$\RegT{\is} = \rfp\TInfOp{\is}$.
\end{lem}
\begin{proof}
To prove $\RegT{\is} \subseteq \rfp\TInfOp{\is}$, let $\tr$ be a regular tree, then $\SubTr{\tr}$ is finite and, by \refToLem{tinfop-post-fixed}, it is a post-fixed point.
Hence, $\tr \in \SubTr{\tr} \subseteq \rfp\TInfOp{\is}$, by \refToProp{below-rfp}, as needed. 

To prove $\rfp\TInfOp{\is} \subseteq \RegT{\is}$, let $X \subseteq \ProofT{\is}$ be a finite post-fixed point of $\TInfOp{\is}$, then we have to show that $X\subseteq \RegT{\is}$. 
Let $\tr \in X$, then, by \refToLem{tinfop-post-fixed}, we have $\SubTr{\tr} \subseteq X$, hence $\SubTr{\tr}$ is finite, that is, $\tr$ is regular, as needed. 
\end{proof}
\noindent Thanks to \refToLem{regtr-rfp} and \refToDef{rational-is}, we trivially get the first equality: $\Reg{\is} = \rtdir(\rfp\TInfOp{\is})$. 

To prove the second equality of \refToThm{rational-sem}, we need a general property of regular trees, which is a stronger version of a result proved in \cite{Dagnino19}. 
To this end, assume a set $A$ and denote by $\Trees{A}$ and $\RegTr{A}$ the sets of all trees and regular trees on $A$, respectively. 
Note that functions $\fun{\dsubtrn}{\Trees{A}}{\wp(\Trees{A})}$ and $\fun{\rt}{\Trees{A}}{A}$, mapping a tree to its direct subtrees and to its root, respectively, are well-defined and restrict to $\RegTr{A}$, because subtrees of a regular tree are regular as well. 
Basically, we show that, starting from a graph structure on a subset of $A$, for each node of the graph there is a unique way to construct a tree coherent with the graph structure, and, moreover, if this subset is finite, all the constructed trees are regular. 
In this context a graph is a function $\fun{g}{X}{\wp(X)}$, modelling the adjacency function, that is, $X$ is the set of nodes and, for all $x \in X$, $g(x)$ is the set of adjacents of $x$. 

\begin{thm} \label{thm:universal-tree}
Let $\fun{g}{X}{\wp(X)}$ be a function and $\fun{v}{X}{A}$ be an injective function. 
Then, there exists a unique function $\fun{p}{X}{\Trees{A}}$ such that the following diagram commutes: 
\[\xymatrix{
X \ar[r]^{p} \ar[d]_{\Pair{g}{v}} & \Trees{A} \ar[d]^{\Pair{\dsubtrn}{\rt}} \\
\wp(X) \times A \ar[r]^{\img{p} \times \id{A}} & \wp(\Trees{A}) \times A
}\]
furthermore, if $X$ is finite, then $p$ corestricts to $\RegTr{A}$, that is, the following diagram commutes: 
\[\xymatrix{
X \ar[r]^{p} \ar[d]_{\Pair{g}{v}} & \RegTr{A} \ar[d]^{\Pair{\dsubtrn}{\rt}} \\
\wp(X) \times A \ar[r]^{\img{p} \times \id{A}} & \wp(\RegTr{A}) \times A
}\]
Finally, $p$ is injective. 
\end{thm}
\begin{proof}
For all $x\in X$, we define the set $L_{x,n}$ of paths of length $n$ starting from $x$ and the set $L_x$ of all paths starting from $x$ as follows: 
\begin{align*}
L_x = \bigcup_{n\in \N} L_{x,n}   &&
\begin{array}{rl} 
L_{x,0}   =& \{\EList\} \\
L_{x,n+1} =& \displaystyle\bigcup_{y \in g(x)} \{ v(y) \alpha \mid \alpha \in L_{y,n}\} 
\end{array} 
\end{align*} 
Trivially we have, for all $x \in X$, $L_x \subseteq \List{A}$. 
We show, by induction on $n$, that for all $n\in \N$, $x\in X$, $\alpha\in\List{A}$ and $a\in A$, if $\alpha a \in L_{x,n+1}$ then $\alpha \in L_{x,n}$. 
\begin{description} 
\item [Base] Since $\alpha a \in L_{x,1}$, we have $\alpha = \EList \in L_{x,0}$, as needed. 
\item [Induction] We prove the thesis for $n+1$. 
Since $\alpha a \in L_{x,n+2}$, by definition of $L_{x,n+2}$, we have $\alpha = v(y)\beta$, for some $y\in g(x)$, and $\beta a \in L_{y,n+1}$. 
By induction hypothesis, we get $\beta\in L_{y,n}$, then, by definition of $L_{x,n+1}$, we get $\alpha = v(y)\beta \in L_{x,n+1}$, as needed. 
\end{description} 
This implies that $L_x$ is prefix-closed, thus a tree language, and so 
$\Pair{x}{L_x}$ is a tree on $A$. 
We define $p(x) = \Pair{x}{L_x}$. 

To prove that the diagram commutes, we have to show that, for all $x \in X$ and $\tr \in \Trees{A}$, $\rt(p(x)) = v(x)$, which is true by construction of $p$,  and  $\tr \in \dsubtr{p(x)}$ iff $\tr = p(y)$ for some $y \in g(x)$. 
First of all, note that, for all $y\in g(x)$ and $\alpha\in \List{A}$, we have $v(y)\alpha\in L_x$ iff $\alpha\in L_y$: 
if $\alpha\in L_y$ then $\alpha\in L_{y,n}$, for some $n\in\N$, thus $v(y)\alpha\in L_{x,n+1}\subseteq L_x$, and, 
if $v(y)\alpha\in L_x$ then \mbox{$v(y)\alpha\in L_{x,n+1}$}, for some $n\in \N$, thus there is $z\in g(x)$ such that $v(z)=v(y)$ and $\alpha\in L_{z,n}\subseteq L_z$, but, since $v$ is  injective, we get $z = y$ and so $\alpha\in L_y$. 
From this fact we immediately get that $p(y)\in \dsubtr{p(x)}$, for all $y\in g(x)$. 
On the other hand, if $\tr\in\dsubtr{p(x)}$, then $\tr = \subtr{p(x)}{a}$, for some $a\in A$, that is,  
$\tr = \Pair{a}{\{ \alpha \in \List{A} \mid a\alpha \in L_x \}}$. 
In particular, we have $a\in L_{x,1}\subseteq L_x$, hence $a = v(y)$, for some $y\in g(x)$. 
Therefore, again thanks to the fact above, we get $\tr = \Pair{v(y)}{L_y} = p(y)$, as needed. 

To prove uniqueness, consider a function $\fun{q}{X}{\Trees{A}}$ making the diagram commute. 
Then, $\rt(q(x)) = v(x) = \rt(p(x))$, hence we have only to show that $\TrNodes{q(x)} = L_x$. 
Therefore, we prove by induction on $\alpha \in \List{A}$ that, for all $x \in X$,  $\alpha \in \TrNodes{q(x)}$ iff $\alpha \in L_x$. 
\begin{description} 
\item [Base] The thesis is trivial. 
\item [Induction] We prove the thesis for $a\alpha$. 
We have $a\alpha \in \TrNodes{q(x)}$ iff $\alpha\in \TrNodes{\subtr{q(x)}{a}}$ and, since the diagram commutes, hence $\subtr{q(x)}{a} = q(y)$, for some $y\in g(x)$, 
this is equivalent to $a = \rt(\subtr{q(x)}{a}) = \rt(q(y)) = v(y)$ and $\alpha\in \TrNodes{q(y)}$, for some $y\in g(x)$. 
By induction hypothesis, this is equivalent to $a = v(y)$ and $\alpha \in L_y$, which is equivalent to $a\alpha\in L_x$. 
\end{description} 

To prove that $p$ corestricts to $\RegTr{A}$, we just have to show that, if $X$ is finite, then $p(x)$ is regular, for all $x \in X$. 
We prove, by induction on $\alpha$,  that, for all $\alpha \in \TrNodes{p(x)}$, there exists $y \in X$ such that $\subtr{p(x)}{\alpha} = p(y)$. 
\begin{description}
\item [Base] We have $\subtr{p(x)}{\EList} = p(x)$, as needed. 
\item [Induction] We prove the thesis for $\beta a$. 
We have $\subtr{p(x)}{\beta a} = \subtr{(\subtr{p(x)}{\beta})}{a}$ and by induction hypothesis, there is $z \in X$ such that $\subtr{p(x)}{\beta} = p(z)$. 
Therefore, we have $\subtr{p(x)}{\alpha} = \subtr{p(z)}{a} \in \dsubtr{p(z)}$, hence, since the diagram commutes, there is $y \in g(z) \subseteq X$ such that $\subtr{p(z)}{a} = p(y)$, as needed. 
\end{description}

Finally, we note that $p$ is injective: if $p(x) = p(y)$ then $v(x) = \rt(p(x)) = \rt(p(y)) = v(y)$, hence $ x = y$ because $v$ is injective. 
\end{proof}
This result looks like a final coalgebra theorem for the functor mapping a set $X$ to $\wp(X)\times A$, but it is not. 
Indeed, such a functor cannot have a final coalgebra, because this would imply the existence of a bijection between a set $Z$, the carrier of such a final coalgebra, and $\wp(Z)\times A$, which is not possible for cardinality reasons. 
Here we manage to have a unique function making the diagram commute, 
because we additionally require the function $\fun{v}{X}{A}$ to be injective. 

We can now prove \refToThm{rational-sem}.

\begin{proofOf}{\refToThm{rational-sem}} 
By \refToLem{regtr-rfp} we get $\rtdir(\RegT{\is}) = \rtdir(\rfp\TInfOp{\is})$. 
Recall that in $\wp(\ProofT{\is})$, compact elements are finite subsets, hence the set of all compact elements is $\finwp(\ProofT{\is})$. 
Then, by definition of the rational fixed point and since $\rtdir$ preserves arbitrary unions (it is a left adjoint), we get 
$\rtdir(\rfp \TInfOp{\is}) = \bigcup \{ \rtdir(X) \mid X  \in \finwp(\ProofT{\is}) \mbox{ and }  X \subseteq \TInfOp{\is}(X)\}$. 
Hence, if $X\in \finwp(\ProofT{\is})$ and $X \subseteq \TInfOp{\is}(X)$,  $\rtdir(X)$ is obviously finite and, by \refToCor{infop-tinfop}, it is also post-fixed. 
Therefore, by definition of the rational fixed point, we get $\rtdir(X) \subseteq \rfp \InfOp{\is}$, and this proves $\rtdir(\rfp\TInfOp{\is}) \subseteq \rfp \InfOp{\is}$. 

To conclude the proof, we show that $\rfp\InfOp{\is} \subseteq \rtdir(\RegT{\is})$. 
To this end, we just have to  prove that, given a finite set $X \in \finwp(\universe)$ such that $X \subseteq \InfOp{\is}(X)$, each judgement $\judg \in X$ has a regular proof tree. 
Since $X \subseteq \InfOp{\is}(X)$, for each $\judg \in X$, there is $\prem_\judg \subseteq X$ such that $\RulePair{\prem_\judg}{\judg} \in \is$. 
Hence, applying \refToThm{universal-tree}, where $g$ maps $\judg$ to $\prem_\judg$ and $v$ is the restriction of the identity on $\universe$ to $X$, we get an injective function $\fun{p}{X}{\RegTr{\universe}}$. 
We have still to prove that $p(\judg)$ is a proof tree. 
It is easy to prove by induction on $\alpha$ that, for all $\alpha \in \TrNodes{p(\judg)}$, there is $\judg' \in X$ such that $\subtr{p(\judg)}{\alpha} = p(\judg')$. 
Therefore, 
$\RulePair{\rtdir(\dsubtr{\subtr{p(\judg)}{\alpha}})}{\rt(\subtr{p(\judg)}{\alpha})} = 
 \RulePair{\rtdir(\dsubtr{p(\judg')})}{\rt(p(\judg'))} = 
 \RulePair{\prem_{\judg'}}{\judg'}$, because $\dsubtrn\circ p = \img{p} \circ g$ and $\rt \circ p = v$, by \refToThm{universal-tree}. 
This proves that $p(\judg)$ is a proof tree in $\is$, as $\RulePair{\prem_{\judg'}}{\judg'} \in \is$ by hypothesis. 
\end{proofOf}

\section{An inductive characterization} \label{sect:cycle} 

Although the regular interpretation is essentially coinductive, as it allows non-well-founded derivations, 
it has an intrinsic finite nature, because it requires proof trees to be regular, that is, to have only finitely many subtrees. 
Given this finiteness, a natural question is the following: 
is it possible to find a finite presentation of derivations for judgements belonging to the regular interpretation? 
f a

In this section we show this is the case, by providing an inductive characterization of the regular interpretation. 

The idea behind such an inductive characterisation is simple. 
Regular trees are basically cyclic structures. 
Usually, to deal with cyclic structures inductively, we need to use auxiliary structures to detect cycles, to ensure termination. 
For instance, in order to perform a visit of a graph, we detect cycles by marking already encountered nodes. 
The inductive characterization \EZ{described below models} such cycle detection mechanism in an abstract and canonical way, in the general setting of inference systems.
The idea is the following: 
during the proof, we keep track of already encountered judgements and, if we find again the same judgement, we can use it as an axiom. 

This approach is intuitively correct, since in  a regular proof tree there are only finitely many subtrees, hence infinite paths \emph{must} contain repeated judgements, and this mechanism is designed precisely to detect such repetitions. 
\EZComm{non lo direi qui: This approach is an abstract version of operational models of programming languages supporting regular corecusion \cite{SimonMBG06,AnconaZ12,AnconaD15}, and can serve as a general framework to check correctness of such models against a more abstract semantics.}

We now formally define the construction and prove its correctness. 
Let $\is$ be a finitary inference system on the universe $\universe$. 
We consider judgements of shape $\LoopJ{\LoopHp}{\judg}$ where $\LoopHp \subseteq \universe$ is a finite set of judgements, called \emph{circular hypotheses}, and $\judg \in \universe$ is a judgement. 
Then, we have the following definition.

\begin{defi} \label{def:cycle-is}
The inference system $\Loopis{\is}$ consists of the following rules: 
\[
\MetaRule{hp}{}{ \LoopJ{\LoopHp}{\judg}}{\judg\in\LoopHp}\BigSpace 
\MetaRule{unfold}{ 
  \LoopJ{\LoopHp\cup\{\judg\}}{\judg_1} 
  \Space\ldots\Space
  \LoopJ{\LoopHp\cup\{\judg\}}{\judg_n}
}{ \LoopJ{\LoopHp}{\judg} }
{\RulePair{\{\judg_1,\ldots,\judg_n\}}{\judg}\in\is}
\]
\end{defi}
Therefore, in the system $\Loopis{\is}$, we have the same rules as in $\is$, that, however, extend the set of circular hypotheses by adding the conclusion of the rule as an hypothesis in the premises. 
Furthermore, $\Loopis{\is}$ has also an additional axiom that allows to apply circular hypotheses. 

\EZ{The correctness of the construction in \refToDef{cycle-is} is expressed by the fact that,} a judgement $\judg$ has a regular proof tree in $\is$ if and only if it has a finite derivation in $\Loopis{\is}$ without circular \EZ{hypotheses}, as formally stated by the next theorem. 

\begin{thm} \label{thm:cycle-is}
$\validInd{\Loopis{\is}}{\LoopJ{\emptyset}{\judg}}$ iff $\validReg{\is}{\judg}$.
\end{thm}

We prove a more general version of the theorem. 
First of all, if $X \subseteq \universe$, we denote by $\Extendis{\is}{X}$ the system obtained from $\is$ by adding an axiom for each element of $X$, hence we have
$\InfOp{\Extendis{\is}{X}}(Y) = \InfOp{\is}(Y) \cup X$, for all $Y\subseteq \universe$. 
Then, the left-to-right implication of \refToThm{cycle-is} is an immediate consequence of the following lemma: 
\begin{lem} \label{lem:cycle-sound}
If $\validInd{\Loopis{\is}}{\LoopJ{\LoopHp}{\judg}}$ then $\validReg{\Extendis{\is}{\LoopHp}}{\judg}$. 
\end{lem}
\begin{proof}
The proof is by induction on rules in $\Loopis{\is}$. 
There are two types of rules, hence we distinguish two cases:
\begin{description}
\item [\rn{hp}] we have $\judg \in \LoopHp$, hence $\RulePair{\emptyset}{\judg} \in \Extendis{\is}{\LoopHp}$, thus $\validReg{\Extendis{\is}{\LoopHp}}{\judg}$. 
\item [\rn{unfold}] we have a rule $\RulePair{\{\judg_1,\ldots,\judg_n\}}{\judg} \in \is$ and, by induction hypothesis, we get $\validReg{\Extendis{\is}{\LoopHp\cup\{\judg\}}}{\judg_i}$, for all $i \in 1..n$. 
Since $\Reg{\Extendis{\is}{\LoopHp\cup\{\judg\}}}$ is a rational fixed point (\refToThm{rational-sem}), for all $i \in 1..n$, there is a finite set $X_i$ such that 
$\judg_i \in X_i \subseteq \InfOp{\Extendis{\is}{\LoopHp\cup\{\judg\}}}(X_i) = \InfOp{\Extendis{\is}{\LoopHp}}(X_i) \cup \{\judg\}$. 
Set $X = \bigcup_{i=1}^n X_i$, then $X$ is finite, $X \subseteq \InfOp{\Extendis{\is}{\LoopHp}}(X) \cup \{\judg\}$, as $\InfOp{\Extendis{\is}{\LoopHp}}$ is monotone,   and $\judg \in \InfOp{\is}(X)$, because, by construction, $\{\judg_1,\ldots,\judg_n\}\subseteq X$.
Thus, we get  $X \cup \{\judg\} \subseteq \InfOp{\Extendis{\is}{\LoopHp}}(X) \cup \{\judg\} = \InfOp{\Extendis{\is}{\LoopHp}}(X)$, because $\judg \in \InfOp{\is}(X) \subseteq \InfOp{\Extendis{\is}{\LoopHp}}(X)$, hence $X \cup \{\judg\}$ is a post-fixed point of $\InfOp{\Extendis{\is}{\LoopHp}}$, since it is monotone. 
Therefore, since  $X\cup\{\judg\}$ is a finite post-fixed point, 
by \refToProp{below-rfp} and \refToThm{rational-sem}, we get $\validReg{\Extendis{\is}{\LoopHp}}{\judg}$. 
\qedhere
\end{description}
\end{proof}

The proof of the other implication relies on an auxiliary  family of functions indexed over 
finite subsets of $\universe$ and finite graphs $\fun{g}{X}{\wp(X)}$, with $X\in\finwp(\universe)$, mapping 
a judgement $\judg \in  X$ and a subset $S \subseteq X$ to a tree whose nodes are judgements of shape $\LoopJ{\LoopHp'}{\judg'}$. 
This function is recursively defined as follows: 
\[
\toTree{\LoopHp}{g}{\judg}{S} = \begin{cases}
\Rule{}{\LoopJ{\LoopHp\cup S}{\judg}} & \judg \in \LoopHp \cup S \\[3ex]
\Rule{
  \toTree{\LoopHp}{g}{\judg_1}{S\cup\{\judg\}} \Space\ldots\Space \toTree{\LoopHp}{g}{\judg_n}{S\cup\{\judg\}} 
}{ \LoopJ{\LoopHp\cup S}{\judg} }  &  
  \begin{array}{l} 
    \judg \notin \LoopHp\cup S \\
    g(\judg) = \{\judg_1,\ldots,\judg_n\}  
  \end{array} 
\end{cases} 
\]
The function $\toTreen{\LoopHp}$ enjoys the following properties: 
\begin{prop} \label{prop:toTree-total}
For all $\LoopHp\in\finwp(\universe)$, $\fun{g}{X}{\wp(X)}$ with $X \in \finwp(\universe)$, $\judg \in X$ and $S\subseteq X$, 
$\toTree{\LoopHp}{g}{\judg}{S}$ is defined. 
\end{prop} 
\begin{proof}
Denote by $c(S)$ the cardinality of the set $X\setminus (\LoopHp \cup S)$. 
We prove that, for all $n \in \N$ and $S\subseteq X$, if $c(S) = n$ then $\toTree{\LoopHp}{g}{\judg}{S}$ is defined. 
The proof is by induction on $n$. 
\begin{description}
\item [Base] If $c(S) = n = 0$, then $X \subseteq \LoopHp \cup S$, hence $\judg \in \LoopHp \cup S$ hence $\toTree{\LoopHp}{g}{\judg}{S} = \Rule{}{\LoopJ{\LoopHp\cup S}{\judg}}$. 
\item [Induction] If $c(S) = n+1$, if $\judg \in \LoopHp \cup S$ then $\toTree{\LoopHp}{g}{\judg}{S}$ is defined as before; 
otherwise, we have $\judg \notin \LoopHp\cup S$ and, if $g(\judg) = \{\judg_1,\ldots,\judg_k\}$, then, for all $i \in 1..n$,  $\toTree{\LoopHp}{g}{\judg_i}{S\cup\{\judg\}} = \tr_i$ by induction hypothesis, as $c(S\cup\{\judg\}) = n$ since $\judg \notin S$, 
hence $\toTree{\LoopHp}{g}{\judg}{S} = \Rule{\tr_1 \Space\ldots\Space \tr_n}{\LoopJ{\LoopHp\cup S}{\judg}}$, as needed. 
\qedhere
\end{description}
\end{proof}

\begin{prop} \label{prop:toTree-valid}
For all $\LoopHp\in\finwp(\universe)$, $\fun{g}{X}{\wp(X)}$ with $X \in \finwp(\universe)$, $\judg \in X$ and $S\subseteq X$, 
if $\RulePair{g(\judg')}{\judg'} \in \is$, for all $\judg' \in X\setminus\LoopHp$, then 
$\toTree{\LoopHp}{g}{\judg}{S}$ is a finite proof tree for $\LoopJ{\LoopHp\cup S}{\judg}$ in $\Loopis{\is}$.
\end{prop}
\begin{proof}
The proof is a straightforward induction on the definition of $\toTreen{\LoopHp}$. 
\end{proof}

We can now prove the following lemma, which concludes the proof of \refToThm{cycle-is}. 
\begin{lem} \label{lem:cycle-complete}
If $\validReg{\Extendis{\is}{\LoopHp}}{\judg}$ then $\validInd{\Loopis{\is}}{\LoopJ{\LoopHp}{\judg}}$. 
\end{lem}
\begin{proof}
If $\judg \in \Reg{\Extendis{\is}{\LoopHp}}$, since $\Reg{\Extendis{\is}{\LoopHp}} = \rfp\InfOp{\Extendis{\is}{\LoopHp}}$ (\refToThm{rational-sem}), we have that there exists a finite set $X \subseteq \universe$ such that $\judg \in X \subseteq \InfOp{\Extendis{\is}{\LoopHp}}(X) = \InfOp{\is}(X)\cup\LoopHp$. 
Then, for each $\judg' \in X\setminus\LoopHp$, there is $\prem_{\judg'} \subseteq X$ such that $\RulePair{\prem_{\judg'}}{\judg'} \in \is$. 
Define $\fun{g}{X}{\wp(X)}$ by $g(\judg') = \prem_{\judg'}$, if $\judg'\in X\setminus\LoopHp$, and $g(\judg') = \emptyset$ otherwise. 
Therefore, by \refToProp{toTree-total}, $\toTree{\LoopHp}{g}{\judg}{\emptyset}$ is defined and, by \refToProp{toTree-valid}, it is a finite proof tree for $\LoopJ{\LoopHp}{\judg}$ in $\Loopis{\is}$, hence we get $\validInd{\Loopis{\is}}{\LoopJ{\LoopHp}{\judg}}$, as needed.
\end{proof}

We conclude the section by discussing \EZ{a} more operational aspect of \refToDef{cycle-is}. 
In this definition,  we aimed at being as liberal as possible, 
hence, the two types of rules are not mutually exclusive: for a judgement $\LoopJ{\LoopHp}{\judg}$ with $\judg \in \LoopHp$ we can either apply the circular hypothesis or use a rule from $\is$. 
Since \EZ{here} we are \EZ{only} interested in derivability, this aspect is not that relevant, however, it becomes more interesting from an algorithmic perspective. 
Indeed, we can consider an alternative definition, where 
we allow the second type of rule only if $\judg \notin \LoopHp$, 
in other words, we apply circular hypotheses as soon as we can. 

In this way we would have less valid proof trees in $\Loopis{\is}$, but the set of derivable judgements remains the same. 
Indeed, \refToLem{cycle-sound} ensures soundness also of the deterministic version, because any proof tree in the deterministic version is also a proof tree in the non-deterministic one. 
On the other hand,  \refToLem{cycle-complete} ensures completeness of the deterministic version, because the functions $\toTreen{\LoopHp}$ build a proof tree in the deterministic version, since they perform the additional check $\judg \notin \LoopHp$ to apply rules from $\is$.

\section{Regular reasoning} 
\label{sect:reasoning}

In this section we discuss proof techniques for regular reasoning, which \EZ{can be defined} thanks to the results proved in \refToSect{equiv} and \refToSect{cycle}.

Let $\is$ be a finitary inference system on the universe $\universe$. 
We are typically interested in comparing the \EZ{regular} interpretation of $\is$ to a set of judgements $\Spec \subseteq \universe$ (\emph{specification}). 
In particular, we focus on two properties:
\begin{description}
\item [soundness] all derivable judgements belong to $\Spec$, that is, $\Reg{\is} \subseteq \Spec$, 
\item [completeness] all judgements in $\Spec$ are derivable, that is, $\Spec \subseteq \Reg{\is}$. 
\end{description}

For completeness proofs, we rely on the fixed point characterization of $\Reg{\is}$ (\refToThm{rational-sem}). 
Indeed, since $\Reg{\is} = \rfp\InfOp{\is}$, from \refToProp{below-rfp} we get a proof principle, \EZ{which we call} the \emph{regular coinduction principle}, \EZ{expressed by} the following proposition: 

\begin{prop}[Regular coinduction] \label{prop:reg-coind} 
Let $\Spec \subseteq \universe$ be a set of judgements, then 
if, for all $\judg \in \Spec$, there is a finite set $X \subseteq \universe$ such that 
\begin{itemize}
\item $X \subseteq \InfOp{\is}(X)$, and 
\item $\judg \in X$, 
\end{itemize} 
then, $\Spec \subseteq \Reg{\is}$. 
\end{prop}
\begin{proof}
Immediate from \refToProp{below-rfp}. 
\end{proof}
This is very much like the usual coinduction principle, but it additionally requires $X$ to be finite. 
The condition $X \subseteq \InfOp{\is}(X)$ can be equivalently expressed  as follows: 
for all $\judg \in X$, there is a rule $\RulePair{\prem}{\judg} \in \is$ such that $\prem \subseteq X$. 

\begin{exa}
To show how proofs by regular coinduction work, as first example we consider the introductory one: the definition of the judgement $\allPos{s}$, where $s$ is  a stream of natural numbers, which, intuitively, should hold when $s$ is positive, that is, it contains only  positive elements.
We report here the inference system $\posis$ defining this predicate: 
\[ \Rule{\allPos{s}}{\allPos{x\cons s}}\ x > 0 \] 
The specification $\posSpec$ is the set of judgements $\allPos{s}$, where $s$ is rational,   meaning that it has finitely many different substreams,    and positive. 
Then, the completeness statement is the following:
\begin{quote}
If $s$ is rational and positive, then $\validReg{\posis}{\allPos{s}}$. 
\end{quote}
To prove the result, let $s$ be a rational stream containing only positive elements and set $X_s = \{ \allPos{s'} \mid s = x_1\cons \ldots \cons x_n\cons s'\}$. 
Clearly, $X_s$ is finite, because $s$ is rational,  and  $\allPos{s} \in X_s$, hence we have only to prove that it is consistent. 
Let $\allPos{s'} \in X_s$, then $s' = x\cons s''$, thus $s = x_1\cons \ldots\cons x_n \cons x \cons s''$, and so $x > 0$, because it is an element of $s$, and $\allPos{s''} \in X_s$, by definition of $X_s$, and this proves that $X_s$ is a post-fixed point. 
Therefore, by the regular coinduction principle we get the thesis. 
\end{exa}

Let us now focus on the soundness property. 
If we interpreted $\is$ inductively, we would prove soundness by induction on rules, but in the regular case this technique is not available, since it is unsound. 
However, in \refToThm{cycle-is}, we proved that $\judg \in \Reg{\is}$  if and only if $\LoopJ{\emptyset}{\judg}$ is derivable in $\Loopis{\is}$, which is interpreted inductively. 
Therefore, we can exploit the induction principle associated with $\Loopis{\is}$ to prove soundness, as the following proposition states: 

\begin{prop} \label{prop:sound-ind}
Let $\Spec \subseteq \universe$ be a set of judgements, then, 
if there is a family $(\Spec_\LoopHp)_{\LoopHp \in \finwp(\universe)}$  such that $\Spec_\LoopHp \subseteq \universe$ and $\Spec_\emptyset \subseteq  \Spec$, and, 
for all $\LoopHp \in \finwp(\universe)$, 
\begin{itemize}
\item $\LoopHp \subseteq \Spec_\LoopHp$, and 
\item for all rules $\RulePair{\prem}{\judg} \in \is$, if $\prem \subseteq \Spec_{\LoopHp \cup\{\judg\}}$ then $\judg \in \Spec_\LoopHp$, 
\end{itemize}
then $\Reg{\is} \subseteq \Spec$. 
\end{prop}
\begin{proof}
By induction on $\Loopis{\is}$, we immediately get that $\validInd{\Loopis{\is}}{\LoopJ{\LoopHp}{\judg}}$ implies $\judg \in \Spec_\LoopHp$. 
Therefore, if $\judg \in \Reg{\is}$, by \refToThm{cycle-is}, we have $\validInd{\Loopis{\is}}{\LoopJ{\emptyset}{\judg}}$, hence $\judg \in \Spec_\emptyset \subseteq  \Spec$. 
\end{proof}
In other words, given a specification $\Spec \subseteq \universe$, to prove soundness we have first to generalize the specification to a family of specifications,
indexed over finite sets of judgements, in order to take into account circular hypotheses. 
Then, we reason by induction on rules in the equivalent inductive system (see \refToDef{cycle-is}) and, since $\Spec_\emptyset \subseteq \Spec$, we get soundness. 

\begin{exa}
We illustrate the technique again on the definition of $\allPos{s}$. 
The soundness statement is the following:
\begin{quote}
If $\validReg{\posis}{\allPos{s}}$, then $s$ is rational and positive. 
\end{quote}
The first step is to generalize the specification to a family of sets $\posSpec_\LoopHp$, indexed over finite subsets of judgements $\LoopHp$. 
\begin{quote}
$\allPos{s} \in \posSpec_\LoopHp$ iff 
either $s$ is rational and positive, 
or $s = x_1\cons\ldots \cons x_n\cons s'$ with $x_i > 0$, for all $i \in 1..n$, and $\allPos{s'} \in \LoopHp$. 
\end{quote}
It is easy to see that $\posSpec_\emptyset \subseteq \posSpec$ and, for all $\LoopHp \in \finwp(\universe)$, $\LoopHp \subseteq \posSpec_\LoopHp$, by definition of $\posSpec_\LoopHp$. 
Hence, we have only to check that the sets $\posSpec_\LoopHp$ are closed with respect to the rule, as formulated in \refToProp{sound-ind}. 

Let us assume $\allPos{s} \in \posSpec_{\LoopHp'}$, with $\LoopHp' = \LoopHp \cup \{\allPos{x\cons s}\}$ and $x > 0$. 
We have the following cases:
\begin{itemize}
\item If $s$ is rational and positive, this is true for $x\cons s$ as well, because $x>0$ by hypothesis. 
\item If $s = x_1\cons\ldots\cons x_n\cons s'$ with $x_i >0$, for all $i \in 1..n$, and $\allPos{s'} \in \LoopHp'$, then, 
if \mbox{$\allPos{s'} \in \LoopHp$}, since $x\cons s = x\cons x_1\cons \ldots\cons x_n\cons s'$ and $x>0$, we have the thesis; 
if $s' = x\cons s$ then $x\cons s = x\cons x_1\cons\ldots\cons x_n \cons x\cons s$, thus it is rational and positive, as $x>0$. 
\qedhere
\end{itemize}
\end{exa} 

We now consider a more complex example: the definition of the distance in a graph (see page \pageref{page:ex-dist}), proving it is sound and complete with respect to the expected meaning. 
\begin{exa}
For the reader's convenience, we report here the rules defining this judgement:
\begin{small}
\[ 
\MetaRule{empty}{}{ \gdist{\Gr}{\node}{\node}{0} }{}  
\BigSpace 
\MetaRule{adj}{
  \gdist{\Gr}{\node_1}{\anode}{\delta_1} \Space \ldots \Space \gdist{\Gr}{\node_n}{\anode}{\delta_n}
}{ \gdist{\Gr}{\node}{\anode}{1 + \min \{\delta_1, \ldots, \delta_n\} }  }{ 
\node \ne \anode \\ 
\Gr(\node) = \{\node_1,\ldots,\node_n\} 
}
\]
\end{small}
We denote by $\distis$ the above inference system. 
We recall for the reader's convenicence a few definitions we need in the proof.  
Let us assume a graph $\fun{\Gr}{\Nodes}{\wp(\Nodes)}$. 
An edge in $\Gr$ is a pair $\Pair{\node}{\anode}$ such that $\anode \in \Gr(\node)$, often written $\node\anode$. We denote by $\Edges$ the set of edges in $\Gr$. 
A path from $\node_0$ to $\anode_n$ in $\Gr$ is a non-empty finite sequence of nodes $\alpha = \node_0\ldots\node_n$ with $n\ge 0$, such that, for all $i \in 1..n$, $\node_{i-1}\node_i \in \Edges$. 
The empty path starting from the node $\node$ to itself is the sequence $\node$.
If $\alpha$ is a path in $\Gr$, then we denote by $\pthLen{\alpha}$ the length of $\alpha$, that is, the number of edges in $\alpha$, and we write $\node\in\alpha$ when  the node $\node$ occurs in $\alpha$, that is, the path $\alpha$ traverses $\node$. 
The distance from a node $\node$ to a node $\anode$, denoted by $\delta(\node,\anode)$,  is the least length of a path from $\node$ to $\anode$, that is, 
$\delta(\node,\anode) = \min\{\pthLen{\alpha} \mid \alpha\mbox{ is a path from $\node$ to $\anode$}\}$, 
hence, if there is no path from $\node$ to $\anode$,  $\delta(\node,\anode)  = \min\emptyset =  \infty$. 
We say a path $\alpha = \node_0\ldots\node_n$ is \emph{simple} if it visits every node at most once, that is, $\node_i = \node_j$ implies $i = j$, for all $i,j\in 0..n$. 
Note that the empty path is trivially simple. 
It is also important to note that $\delta(\node,\anode)$ is the least length of a simple  path from $\node$ to $\anode$. 
Then, the specification $\distSpec$ is the set of judgements $\gdist{\Gr}{\node}{\anode}{\delta}$ with $\delta = \delta(\node,\anode)$. 

We can now state that the definition of $\gdist{\Gr}{\node}{\anode}{\delta}$ is sound and complete with respect to the specification $\distSpec$. 
\begin{quote}
$\validReg{\distis}{\gdist{\Gr}{\node}{\anode}{\delta}}$ iff $\delta = \delta(\node,\anode)$. 
\end{quote}

\subsubsection*{Completeness proof}
The proof is by regular coinduction. 
Let us consider a judgement $\gdist{\Gr}{\node}{\anode}{\delta(\node,\anode)}$. 
Let $R_\node \subseteq \Nodes$ be the set of nodes reachable from $\node$ and let us define $X_\node = \{\gdist{\Gr}{\node'}{\anode}{\delta(\node',\anode)}  \mid \node' \in R_\node\}$. 
This set is clearly finite and, moreover,  $\gdist{\Gr}{\node}{\anode}{\delta(\node,\anode)} \in X_\node$, because $\node$ is reachable from itself. 
Hence, we have only to prove that $X_\node$ is a post-fixed point. 
Let $\node' \in R_\node$, then we have to find a rule with conclusion $\gdist{\Gr}{\node'}{\anode}{\delta(\node',\anode)}$ and whose premises are in $X_\node$. 
We have two cases: 
\begin{itemize}
\item If $\node' = \anode$, then $\delta(\node',\anode) = 0$ and so we have the thesis by rule \rn{empty}. 
\item If $\node'\ne \anode$, then 
we have $\delta(\node',\anode) = 1+\min\{\delta(\node'',\anode) \mid \node'' \in \Gr(\node')\}$, hence, 
since $\Gr(\node') \subseteq R_\node$, all the premises $\gdist{\Gr}{\node''}{\anode}{\delta(\node'',\anode)}$, for $\node''\in\Gr(\node')$,  belong to $X_\node$, as needed. 
\end{itemize}

\subsubsection*{Soundness proof} 
To apply \refToProp{sound-ind}, we generalize the specification $\distSpec$ to a family of sets $\distSpec_\LoopHp$, indexed over finite sets of judgements, defined below. 
\begin{quote}
\distStar $\gdist{\Gr}{\node}{\anode}{\delta} \in \distSpec_\LoopHp$ iff there is a set of paths $P$ and a function $\fun{f}{P}{\N\cup\{\infty\}}$ such that 
\begin{enumerate}
\item\label{itm:diststar:1} for all $\alpha \in P$, 
either $\alpha$ goes from $\node$ to $\anode$ and $f(\alpha) = 0$, 
or $\alpha$ goes from $\node$ to $\node'$ and $\gdist{\Gr}{\node'}{\anode}{f(\alpha)} \in \LoopHp$;
\item\label{itm:diststar:2} for each simple path $\beta$ from $\node$ to $\anode$, there is $\alpha \in P$ such that $\beta = \alpha\beta'$; 
\item\label{itm:diststar:3} $\delta = \min\{\pthLen{\alpha}+f(\alpha) \mid \alpha \in P\}$. 
\end{enumerate}
\end{quote}
First, we have to check that  $\distSpec_\emptyset \subseteq \distSpec$.  
Let $\gdist{\Gr}{\node}{\anode}{\delta} \in \distSpec_\emptyset$, then, by \refToItm{diststar:3} of \distStar,  $\delta = \min\{ \pthLen{\alpha}+f(\alpha) \mid \alpha \in P \}$, for some set of paths $P$ and function $\fun{f}{P}{\N\cup\{\infty\}}$. 
Since $\LoopHp$ is empty, by \refToItm{diststar:1} of \distStar, we have that, for all $\alpha \in P$, $\alpha$ is a path from $\node$ to $\anode$ and $f(\alpha) = 0$, hence $\delta(\node,\anode) \le \pthLen{\alpha} + f(\alpha)$ for all $\alpha\in P$ and so $\delta(\node,\anode) \le \delta$. 
To prove the other inequality, let $\beta$ be a simple path from $\node$ to $\anode$, then, 
by \refToItm{diststar:2}  of \distStar,  there is $\alpha \in P$ such that $\beta = \alpha\beta'$, but, by \refToItm{diststar:1} of \distStar, $\alpha$ goes from $\node$ to $\anode$ and $f(\alpha) = 0$, hence $\beta = \alpha$, because it cannot traverse twice $\anode$; 
thus we have $\delta \le \pthLen{\beta} = \pthLen{\alpha}+f(\alpha)$, for any simple path $\beta$, and so $\delta \le \delta(\node,\anode)$. 

The fact that $\LoopHp\subseteq \distSpec_\LoopHp$ is immediate because, if $\gdist{\Gr}{\node}{\anode}{\delta} \in \LoopHp$, then, to get the thesis, it is enough to take as $P$ the set containing only the empty path with $f(\node) = \delta$, which trivially satisfies all conditions in \distStar. 

Then, we have only to check that all sets $\distSpec_\LoopHp$ are closed  with respect to the rules \rn{empty} and \rn{adj}, as formulated in \refToProp{sound-ind}. 
\begin{description}
\item [\rn{empty}] If $\node = \anode$ and $\delta = 0$, then it is enough to take as $P$ the set containing only the empty path, with $f(\node) = 0$. 
\item [\rn{adj}] We have $\node \ne \anode$, $\Gr(\node) = \{\node_1, \ldots, \node_n\}$ and, for all $i \in 1..n$, $\gdist{\Gr}{\node_i}{\anode}{\delta_i} \in \distSpec_{\LoopHp'}$ with $\LoopHp' = \LoopHp\cup\{\gdist{\Gr}{\node}{\anode}{\delta}\}$. 
If $n = 0$, then $\Gr(\node)$ is empty, $\delta = \min\emptyset = \infty$ and there is no path from $\node$ to $\anode$. 
Hence, the thesis follows by taking $P = \emptyset$. 

Then, let us assume $n \ge 1$. 
By hypothesis,  \mbox{$\delta = 1 + \min \{\delta_1,\ldots,\delta_n\} = 1 + \delta_k$,} for some $k\in 1..n$,  since we are considering rule \rn{adj}. 
Since $\gdist{\Gr}{\node_i}{\anode}{\delta_i} \in \distSpec_{\LoopHp'}$, for all $i \in 1..n$, there are $P_i$ and $\fun{f_i}{P_i}{\N\cup\{\infty\}}$ satisfying \distStar, in particular, by \refToItm{diststar:3}, we have that  $\delta_i = \min \{ \pthLen{\alpha} + f_i(\alpha) \mid \alpha \in P_i \}$. 
We define $P$ as the set of paths $\node\alpha$ with $\alpha\in P_i$ such that, if $\alpha$ ends in $\node$, then $f_i(\alpha) \ne \delta$, and 
$\fun{f}{P}{\N\cup\{\infty\}}$   is   defined by $f(\node\alpha) = f_i(\alpha)$ when $\alpha \in P_i$. 
Clearly, $P$ satisfies \refToItm{diststar:1}  of \distStar with respect to $\LoopHp$. 
To check that \refToItm{diststar:2}  holds, let $\beta$ be a simple path from $\node$ to $\anode$, then $\beta = \node\node_i\beta'$, for some $i \in 1..n$. 
Hence, $\node_i\beta'$ is a simple path from $\node_i$ to $\anode$ and $\node \notin \node_i\beta'$, thus, by \refToItm{diststar:2} of \distStar applied to $P_i$, there is $\alpha'\in P_i$ such that $\node_i\beta'=\alpha' \gamma$, 
and $\node\notin \alpha'$, because $\node\notin\node_i\beta'$. 
Therefore, $\node\alpha' \in P$ and $\node\alpha'\gamma = \node\node_i\beta' = \beta$, as needed. 

We now prove \refToItm{diststar:3} of \distStar, that is, $\delta = \min \{ \pthLen{\alpha} + f(\alpha) \mid \alpha \in P \}$. 
Let $\alpha = \node\node_i\alpha' \in P$, for some $i\in 1..n$,  then $\node_i\alpha'\in P_i$, hence, $\delta_k\le \delta_i \le \pthLen{\node_i\alpha'} + f_i(\node_i\alpha')$, thus 
$$\delta = 1+\delta_k \le 1+ \pthLen{\node_i\alpha'} + f_i(\node_i\alpha') = \pthLen{\alpha} + f(\alpha)$$ 
and this implies 
$\delta \le \min\{\pthLen{\alpha}+f(\alpha) \mid \alpha \in P\}$. 
To conclude, we have to prove the other inequality, 
hence we distinguish the following cases:
\begin{itemize}
\item if $\delta_k = \infty$, then $\delta = \infty$ and this proves the thesis, since $\infty \ge x$ for all $x\in \N\cup\{\infty\}$; 
\item otherwise, $\delta_k = \pthLen{\alpha'} + f_k(\alpha')$, for some $\alpha' \in P_k$.
If $\alpha'$ ends in $\node$ and $f_k(\alpha') = \delta$, then $\delta_k = \pthLen{\alpha'} + \delta = \pthLen{\alpha'} + 1 + \delta_k$, which implies $\delta_k = \infty$ that is  absurd. 
Otherwise, $\node\alpha' \in P$ and  $f(\node\alpha') = f_k(\alpha')$, thus 
$$\min\{\pthLen{\alpha}+f(\alpha) \mid \alpha \in P\} \le \pthLen{\node\alpha'}+f(\node\alpha') = 1 + \pthLen{\alpha'}+f_k(\alpha') = 1+\delta_k = \delta$$ 
as needed. 
\end{itemize}
\end{description}
\end{exa}

\section{Flexible regular coinduction} \label{sect:corules}

Infinite derivations are a very powerful tool, which make it possible  to deal with a variety of situations that cannot be handled by only finite derivations. 
However, in some cases, they have an unexpected behaviour, allowing the derivation of intuitively incorrect judgements. 
The same issue affects also regular derivations.
Let us explain this by an example.
Consider the following rules, defining the judgement $\minElem{x}{l}$, where $x$ is an integer and $s$ is a rational stream, stating that $x$ is the minimum of the stream $s$. 
\[
\Rule{ \minElem{y}{s} }{ \minElem{z}{x\cons s} }\,z=\min\{x,y\}  
\]
In \refToFig{min-trees} we report three infinite regular derivations, thus valid for the regular interpretation of the above rules, 
where, however, only the first one is intuitively correct: judgements $\minElem{0}{2\cons 2\cons\ldots}$ and $\minElem{1}{2\cons 2\cons \ldots}$ should not be derivable, as $0$ and $1$ do not belong to the stream. 

\begin{figure}
\[
\Rule{
  \Rule{
    \Rule{\vdots}{ \minElem{2}{2\cons 2\cons 2\cons\ldots} }
  }{ \minElem{2}{2\cons 2\cons 2\cons\ldots} }
}{ \minElem{2}{2\cons 2\cons 2\cons\ldots} }
\BigSpace
\Rule{
  \Rule{
    \Rule{\vdots}{ \minElem{1}{2\cons 2\cons 2\cons\ldots} }
  }{ \minElem{1}{2\cons 2\cons 2\cons\ldots} }
}{ \minElem{1}{2\cons 2\cons 2\cons\ldots} }
\BigSpace
\Rule{
  \Rule{
    \Rule{\vdots}{ \minElem{0}{2\cons 2\cons 2\cons\ldots} }
  }{ \minElem{0}{2\cons 2\cons 2\cons\ldots} }
}{ \minElem{0}{2\cons 2\cons 2\cons\ldots} }
\]
\caption{Some infinite regular derivation for the judgement $\minElem{x}{s}$. }  \label{fig:min-trees}
\end{figure}

Inference systems with corules \cite{AnconaDZ17esop,Dagnino19} have been recently designed precisely to address this issue for the coinductive interpretation, where arbitrary infinite derivations are allowed. 
Beside standard inference rules, they introduce special rules, called \emph{corules}, which allow one to refine the coinductive interpretation,
by filtering out some, undesired, infinite derivations. 
More precisely, an \emph{inference system with corules}, or \emph{generalised inference system}, is a pair $\Pair{\is}{\cois}$ where $\is$ and $\cois$ are inference systems, whose elements are called \emph{rules} and \emph{corules}, respectively.
A corule is also denoted by $\CoRule{\prem}{\conclu}$, very much like a rule, but with a thicker line. 

The semantics of such a pair, denoted by $\FlexCo{\is}{\cois}$, is constructed in two steps: 
\begin{itemize}
\item first, we take the inductive   interpretation of the union $\is\cup\cois$, that is, $\Ind{\is\cup\cois}$, 
\item then, we take the coinductive interpretation of $\is$ restricted to $\Ind{\is\cup\cois}$. 
\end{itemize}
In symbols, we have $\FlexCo{\is}{\cois} = \CoInd{\is\Restrict{\Ind{\is\cup\cois}}}$, where 
$\is\Restrict{X}$ is the inference system obtained from $\is$ by keeping only rules with conclusion in $X\subseteq \universe$. 

In terms of proof trees, $\FlexCo{\is}{\cois}$ is the set of judgements with an arbitrary (finite or not) proof tree in $\is$, whose nodes all have a finite proof tree in $\is {\cup} \cois$. 
\EZComm{tagliato mi sembrava una ripetizione: Essentially, by corules we filter out some, undesired, infinite proof trees.}
In \cite{Dagnino19}, $\FlexCo{\is}{\cois}$ is proved to be an interpretation of $\is$, that is, a fixed point of $\InfOp{\is}$. 

\EZ{In this section, we show that the results previously given for regular coinduction smoothly extend to generalised inference systems.
The technical development in the following is partly repetitive; this could have been avoided by  presenting the results in the generalized framework since the beginning. 
However, to have separation of concerns, we preferred to first give a presentation using only standard notions, limiting to this section non-standard ones.}

We start by defining the regular interpretation of an inference system with corules. 
\begin{defi}\label{def:reg-corules}
Let $\Pair{\is}{\cois}$ be an inference system with corules. 
The \emph{regular interpretation} $\FlexReg{\is}{\cois}$ of $\Pair{\is}{\cois}$ is defined by 
$\FlexReg{\is}{\cois} = \Reg{\is\Restrict{\Ind{\is\cup\cois}}}$. 
\end{defi}
As we will see later in this section (\refToCor{pt-corules}), but it is not difficult to be convinced of it, 
in proof-theoretic terms this is equivalent to say that $\FlexReg{\is}{\cois}$ is the set of judgements with a regular proof tree in $\is$, 
whose nodes all have a finite proof tree in $\is\cup\cois$. 
In this way, we can filter out some, undesired, regular derivations. 
In the following, we will write $\validReg{\Pair{\is}{\cois}}{\judg}$ for $\judg \in \FlexReg{\is}{\cois}$. 

Coming back to the example, using corules, we can provide a correct definition of the judgement $\minElem{x}{s}$ as follows: 
\[
\Rule{ \minElem{y}{s} }{ \minElem{z}{x\cons s} }\,z=\min\{x,y\}  \BigSpace
\CoAxiom{ \minElem{x}{x\cons s} }
\]
The additional constraint, imposed by the coaxiom, allows us to build regular infinite derivation using only judgements $\minElem{x}{s}$ where $x$ belongs to $s$;
thus filtering out the second and third incorrect proof trees in \refToFig{min-trees}, since they involve judgements with no finite derivation using also the coaxiom. 

All the results discussed so far  for the regular interpretation can be smoothly extended to the regular interpretation of an inference system with corules. 
We will now develop all the tecnical machinery needed for this, adapting constructions in \cite{AnconaDZ17esop,Dagnino19} to the regular case. 

\subsection{Bounded rational fixed point}
To construct such a fixed point, we come back to the lattice-theoretic setting of \refToSect{rfp}.
Let us assume an algebraic lattice $\Pair{\lattice}{\order}$.

Let $\fun{\MFun,\aMFun}{\lattice}{\lattice}$ be two functions, 
we write $\MFun\join\aMFun$ for the pointwise join of $\MFun$ and $\aMFun$,
and, for all $z\in \lattice$, $\FMeet{\MFun}{z}$ for the function defined by $\FMeet{\MFun}{z}(x) = \MFun(x)\meet z$. 
It is easy to see that, if $\MFun$ and $\aMFun$ are monotone, then $\MFun\join\aMFun$ is monotone as well, hence, by the Tarski theorem, it has a least fixed point $\lfp(\MFun\join\aMFun)$. 
It is also easy to check that if $z\in\lattice$ is a  pre-fixed point of $\MFun\join\aMFun$, then it is a pre-fixed point of $\MFun$ as well, because $\MFun(z) \order \MFun(z)\join\aMFun(z) \order z$; 
this will be crucial for the following construction. 

We can now define the bounded rational fixed point: 
\begin{defi}\label{def:brfp}
Let $\fun{\MFun}{\lattice}{\lattice}$ be finitary and  $\fun{\aMFun}{\lattice}{\lattice}$ be monotone.
The \emph{rational fixed point bounded by $\aMFun$}, $\brfp{\MFun}{\aMFun}$ is defined by 
\[ \brfp{\MFun}{\aMFun} = \rfp \FMeet{\MFun}{\lfp(\MFun\join\aMFun)} \]
\end{defi}
In other words, $\brfp{\MFun}{\aMFun}$ is the least upper bound of all compact elements \emph{below} the least fixed point of $\MFun\join\aMFun$, that is, 
\[ \brfp{\MFun}{\aMFun} = \lub \{ x \in \Compact{\lattice} \mid x\order\MFun(x),\, x\order \lfp(\MFun\join\aMFun) \} \]
To see \EZ{that} $\brfp{\MFun}{\aMFun}$ is well-defined, that is, it is indeed a fixed point of $\MFun$, we have the following propositions:
\begin{prop}\label{prop:fmeet-finitary}
If $\fun{\MFun}{\lattice}{\lattice}$ is finitary, then, for all $z\in\lattice$, $\FMeet{\MFun}{z}$ is finitary as well.
\end{prop}
\begin{proof}
Let $D\subseteq\lattice$ be a directed set. 
Since $\MFun$ is finitary, it is monotone, hence $\FMeet{\MFun}{z}$ is monotone as well, therefore we get 
$\lub (\img{\MFun}(D)\meet z) \order \MFun\left(\lub D\right) \meet z$. 
To prove the other inequality, it is enough to show that, for any compact element $y\order \MFun\left(\lub D\right) \meet z$, $y\order \lub (\img{\MFun}(D)\meet z)$, because the lattice is algebraic. 
Since $\MFun$ is finitary, we have $\MFun\left(\lub D\right) = \lub \img{\MFun}(D)$. 
We know that $y\order z$ and $y\order \MFun\left(\lub D\right) = \lub \img{\MFun}(D)$ and, since $y$ is compact, there is a finite subset $W\subseteq D$ such that $y\order \lub \img{\MFun}(W)$. 
Since $D$ is directed and $W$ is finite, there is $w\in D$ such that $\lub W\order w$, hence $\lub\img{\MFun}(W) \order \MFun(w) \order \lub \img{\MFun}(D)$, because $\MFun$ is monotone and $w\in D$. 
Therefore, we get $y \order \MFun(w)\meet z \order \lub (\img{\MFun}(D) \meet z)$, as needed. 
\end{proof}

\begin{prop}\label{prop:brfp-wf}
Let $\fun{\MFun}{\lattice}{\lattice}$ be finitary and  $\fun{\aMFun}{\lattice}{\lattice}$ be monotone, then 
$\brfp{\MFun}{\aMFun}$ is a fixed point of $\MFun$. 
\end{prop}
\begin{proof}
Set $z=\lfp(\MFun\join\aMFun)$ and note that $\MFun(z)\order \MFun(z)\join\aMFun(z) = z$, as $z$ is a fixed point of $\MFun\join\aMFun$. 
By \refToProp{fmeet-finitary}, $\FMeet{\MFun}{z}$ is finitary, hence, by \refToDef{brfp} and \refToThm{rfp} we have $\brfp{\MFun}{\aMFun} = \MFun(\brfp{\MFun}{\aMFun})\meet z$,
and from this we derive  $\brfp{\MFun}{\aMFun} \order z$ and $\MFun(\brfp{\MFun}{\aMFun}) \order \MFun(z) \order z$. 
Therefore, we get $\brfp{\MFun}{\aMFun} = \MFun(\brfp{\MFun}{\aMFun}) \meet z = \MFun(\brfp{\MFun}{\aMFun})$, as needed. 
\end{proof}

In \cite{AnconaDZ17esop,Dagnino19} \EZ{the} authors show that the least and the greatest fixed point are instances of the bounded fixed point. 
Analogously, we show that the least and the rational fixed point are instances of the rational fixed point bounded by a function $\aMFun$, that is, they can be recovered for specific choices of $\aMFun$. 
In the following, for all $z\in\lattice$, we write $\fun{K_z}{\lattice}{\lattice}$ for the constant function, that is, $K_z(x) = z$, for all $x\in \lattice$. 
\begin{prop}\label{prop:brfp-lfp-rfp}
Let $\fun{\MFun}{\lattice}{\lattice}$ be a finitary function, then the following hold:
\begin{enumerate}[beginpenalty=99,midpenalty=99]
\item\label{prop:brfp-lfp-rfp:1} $\lfp\MFun = \brfp{\MFun}{K_\bot}$, and 
\item\label{prop:brfp-lfp-rfp:2} $\rfp\MFun = \brfp{\MFun}{K_\top}$.
\end{enumerate}
\end{prop}
\begin{proof}
To prove \ref{prop:brfp-lfp-rfp:1}, note that $\lfp\MFun \order \brfp{\MFun}{K_\bot}$, as $\brfp{\MFun}{K_\bot}$ is a pre-fixed point, and 
$\brfp{\MFun}{K_\bot}\order\lfp\MFun$, as $\lfp(\MFun\join K_\bot) = \lfp\MFun$ and $\brfp{\MFun}{K_\bot} = \MFun(\brfp{\MFun}{K_\bot})\meet \lfp\MFun \order \lfp\MFun$. 
To prove \ref{prop:brfp-lfp-rfp:2}, note that $\lfp(\MFun\join K_\top) = \top$, hence we have $\FMeet{\MFun}{\lfp(\MFun\join K_\top)} = \MFun$, thus $\brfp{\MFun}{K_\top} = \rfp\MFun$, as needed. 
\end{proof}

\subsection{Fixed point semantics} 
Let $\Pair{\is}{\cois}$ be an inference system with corules where $\is$ is finitary. 
We have two goals: first we want to justify the proof-theoretic characterisation provided at the beginning of this section and, 
second, we want to prove that the rational interpretation generated by corules is indeed an interpretation of the first inference system. 

To get the proof-theoretic characterisation, it is enough to observe the following property: 
\begin{prop}\label{prop:pt-corules}
Let $X\subseteq\universe$, then 
$\tr\in\RegT{\is\Restrict{X}}$ iff $\tr\in\RegT{\is}$ and, for all $\alpha\in\TrNodes{\tr}$, $\tr(\alpha)\in X$. 
\end{prop}
\begin{proof}
By definition we have $\is\Restrict{X}\subseteq \is$, hence $\RegT{\is\Restrict{X}} \subseteq \RegT{\is}$, and, 
all rules in $\is\Restrict{X}$ have conlcusion in $X$.
Then, the thesis is immediate since, by definition, all nodes of a proof tree are labelled by the conclusion of some rule. 
\end{proof}

Recall that we have described $\FlexReg{\is}{\cois}$ in proof-theoretic terms  as the set of judgements having a regular proof tree in $\is$, whose nodes all have a finite proof tree in $\is\cup\cois$.
Formally, we have the following corollary:
\begin{cor}\label{cor:pt-corules}
$\validReg{\Pair{\is}{\cois}}{\judg}$ iff there is $\tr\in\RegT{\is}$ such that $\rt(\tr) = \judg$ and, for all $\alpha\in\TrNodes{\tr}$, $\tr(\alpha)\in\Ind{\is\cup\cois}$. 
\end{cor}
\begin{proof}
Set $X=\Ind{\is\cup\cois}$. 
From \refToThm{rational-sem} and \refToDef{reg-corules}, we get 
$\FlexReg{\is}{\cois} = \Reg{\is\Restrict{X}} = \rtdir(\rfp\TInfOp{(\is\Restrict{X})}) = \rfp \InfOp{(\is\Restrict{X})} = \rtdir(\RegT{\is\Restrict{X}})$. 
Applying \refToProp{pt-corules} with $X = \Ind{\is\cup\cois}$,
we get the thesis.
\end{proof}

Towards the second goal, we show that the regular interpretation of $\Pair{\is}{\cois}$  coincides with the rational fixed point of $\InfOp{\is}$ bounded by $\InfOp{\cois}$ (see \refToDef{brfp}), which is an immediate consequence of the following proposition: 
\begin{prop} \label{prop:brfp-corules}
$\FlexReg{\is}{\cois} = \brfp{\InfOp{\is}}{\InfOp{\cois}}$. 
\end{prop}
\begin{proof}
By \refToDef{reg-corules} and \refToThm{rational-sem}, we know that $\FlexReg{\is}{\cois} = \rfp\InfOp{(\is\Restrict{\Ind{\is\cup\cois}})}$. 
By \refToDef{brfp}, we have $\brfp{\InfOp{\is}}{\InfOp{\cois}} = \rfp \FMeet{(\InfOp{\is})}{\lfp(\InfOp{\is}\cup\InfOp{\cois})}$ and, by definition of the inference operator, we have $\InfOp{\is\cup\cois} = \InfOp{\is}\cup\InfOp{\cois}$, hence $\Ind{\is\cup\cois} = \lfp(\InfOp{\is}\cup\InfOp{\cois})$ and $\FMeet{(\InfOp{\is})}{\lfp(\InfOp{\is}\cup\InfOp{\cois})} = \FMeet{(\InfOp{\is})}{\Ind{\is\cup\cois}}$. 
Therefore, as proved in \cite{Dagnino19}, we have $\InfOp{(\is\Restrict{\Ind{\is\cup\cois}})} = \FMeet{(\InfOp{\is})}{\Ind{\is\cup\cois}}$, which implies the thesis. 
\end{proof}
Then, this proposition, together with \refToProp{brfp-wf}, in particular ensures that $\FlexReg{\is}{\cois}$ is indeed a fixed point of $\InfOp{\is}$, that is, an interpretation of $\is$. 

An important property of inference systems with corules is that standard interpretations (the inductive and the coinductive one) are particular cases. 
Analogously, the inductive and the regular interpretations are particular cases of the regular interpretation generated by corules. 
Let us denote by $\is_\universe$ the inference system containing one axiom for each $\judg\in\universe$, that is, 
$\RulePair{\prem}{\judg}\in\is_\universe$ iff $\prem = \emptyset$. 
We have the following proposition:
\begin{prop} \label{prop:particular-cases}
Let $\is$ be an inference system, then 
$\Ind{\is} = \FlexReg{\is}{\emptyset}$ and $\Reg{\is} = \FlexReg{\is}{\is_\universe}$. 
\end{prop}
\begin{proof}
It follows from \refToProp{brfp-lfp-rfp}, because $\FlexReg{\is}{\cois} = \brfp{\InfOp{\is}}{\InfOp{\cois}}$, by \refToProp{brfp-corules}, and we have $\InfOp{\is_\universe}(X) = \universe$ and $\InfOp{\emptyset}(X) = \emptyset$, for all $X\subseteq \universe$. 
\end{proof}
In other words, when the set of corules is empty, we allow only rules with conclusion in $\Ind{\is\cup\emptyset} = \Ind{\is}$, hence we cannot derive anything outside $\Ind{\is}$, 
and, on the other hand, when the set of corules is $\is_\universe$, we do not remove any rule, because $\Ind{\is\cup\is_\universe} = \universe$, thus we get exactly the regular interpretation of $\is$. 

\subsection{Cycle detection for corules}
As the standard regular interpretation, also the regular interpretation of an inference system with corules has a sound and complete algorithm \EZ{to find a derivation for a judgment, if any}, and may not terminate otherwise. 

Let us assume an inference system with corules $\Pair{\is}{\cois}$. 
Since its regular interpretation is defined as the regular interpretation of $\is\Restrict{\Ind{\is\cup\cois}}$, which is the inference system obtained from $\is$ by keeping only rules with conclusion in $\Ind{\is\cup\cois}$, 
we could get an inductive characterisation of $\FlexReg{\is}{\cois}$ by applying the construction in  \refToDef{cycle-is} to the inference system $\is\Restrict{\Ind{\is\cup\cois}}$. 
This provides us with a sound and complete algorithm to \EZ{find a derivation for a judgement which} belongs to $\FlexReg{\is}{\cois}$, which works the same way as the one introduced in \refToSect{cycle}, but, in addition, 
each time we apply the rule \rn{unfold} with $\RulePair{\prem}{\judg}\in\is$, we have to check that $\judg\in\Ind{\is\cup\cois}$. 
However, we will see that  this additional check is necessary only to apply circular hypotheses, thus defining a cleaner procedure. 

To this end we construct the inference system $\Loopcois{\is}{\cois}$ as follows: 
\begin{defi} \label{def:cycle-cois}
The inference system $\Loopcois{\is}{\cois}$ consists of the following rules: 
\[
\MetaRule{b-hp}{}{ \LoopJ{\LoopHp}{\judg}}{\judg\in\LoopHp \\ \judg \in \Ind{\is\cup\cois}}\BigSpace 
\MetaRule{b-unfold}{ 
  \LoopJ{\LoopHp\cup\{\judg\}}{\judg_1} 
  \Space\ldots\Space
  \LoopJ{\LoopHp\cup\{\judg\}}{\judg_n}
}{ \LoopJ{\LoopHp}{\judg} }
{\RulePair{\{\judg_1,\ldots,\judg_n\}}{\judg}\in\is}
\]
\end{defi}
This definition is basically the same as \refToDef{cycle-is}, except the additional side condition in rule \rn{b-hp} $\judg\in\Ind{\is\cup\cois}$, which enforces the additional check. 
We have the following fundamental properties:

\begin{prop}\label{prop:loop-bound}
If $\validInd{\Loopcois{\is}{\cois}}{\LoopJ{\LoopHp}{\judg}}$ then $\judg\in\Ind{\is\cup\cois}$.
\end{prop}
\begin{proof}
By induction on rules of $\Loopcois{\is}{\cois}$: 
the case for rule \rn{b-hp} is trivial, 
for the rule \rn{b-unfold}, by \refToDef{cycle-cois}, we have a rule $\RulePair{\{\judg_1,\ldots,\judg_n\}}{\judg}\in\is$ and, by induction hypothesis, we know that $\judg_k\in\Ind{\is\cup\cois}$, for all $k\in 1..n$, hence $\judg\in\Ind{\is\cup\cois}$, as $\Ind{\is\cup\cois}$ is closed with respect to $\is$. 
\end{proof}

\begin{prop} \label{prop:loop-cois}
If $\LoopHp\subseteq\Ind{\is\cup\cois}$, then 
$\validInd{\Loopcois{\is}{\cois}}{\LoopJ{\LoopHp}{\judg}}$ iff $\validInd{\Loopis{\is\Restrict{\Ind{\is\cup\cois}}}}{\LoopJ{\LoopHp}{\judg}}$. 
\end{prop}
\begin{proof}
The proof of the left-to-right implication is by induction on rules in $\Loopcois{\is}{\cois}$. 
\begin{description}
\item [\rn{b-hp}] By hypothesis $\judg\in\Ind{\is\cup\cois}$, then the thesis follows by rule \rn{hp}. 
\item [\rn{b-unfold}] By \refToDef{cycle-cois}, we have a rule $\RulePair{\{\judg_1,\ldots,\judg_n\}}{\judg} \in \is$ and by \refToProp{loop-bound} we have $\judg\in\Ind{\is\cup\cois}$, hence $\LoopHp\cup\{\judg\}\subseteq \Ind{\is\cup\cois}$ and $\RulePair{\{\judg_1,\ldots,\judg_n\}}{\judg} \in \is\Restrict{\Ind{\is\cup\cois}}$. 
Therefore, by induction hypothesis, we get $\validInd{\Loopis{\is\Restrict{\Ind{\is\cup\cois}}}}{\LoopJ{\LoopHp\cup\{\judg\}}{\judg_k}}$, for all $k\in 1..n$, then the thesis follows by rule \rn{unfold}. 
Therefore, we get the thesis applying rule \rn{unfold}. 
\end{description}
The proof of the right-to-left implication is by induction on rules in $\Loopis{\is\Restrict{\Ind{\is\cup\cois}}}$.
\begin{description}
\item [\rn{hp}] Immediate by rule \rn{b-hp}, as $\judg\in\LoopHp\subseteq\Ind{\is\cup\cois}$. 
\item [\rn{unfold}] By \refToDef{cycle-is}, we have a rule  $\RulePair{\{\judg_1,\ldots,\judg_n\}}{\judg}\in \is\Restrict{\Ind{\is\cup\cois}}\subseteq \is$, hence $\judg\in\Ind{\is\cup\cois}$, and so $\LoopHp\cup\{\judg\}\subseteq\Ind{\is\cup\cois}$. 
Therefore, by induction hypothesis, we get $\validInd{\Loopcois{\is}{\cois}}{\LoopJ{\LoopHp\cup\{\judg\}}{\judg_k}}$, for all $k\in 1..n$, then the thesis follows by rule \rn{b-unfold}. 
\qedhere
\end{description}
\end{proof}

Then, we get the following result, proving that the inductive characterisation is correct, that is, sound and complete, with respect to the regular interpretation of $\Pair{\is}{\cois}$. 
\begin{cor} \label{cor:loop-cois}
$\validInd{\Loopcois{\is}{\cois}}{\LoopJ{\emptyset}{\judg}}$ iff $\validReg{\Pair{\is}{\cois}}{\judg}$. 
\end{cor}
\begin{proof}
It is immediate by \refToProp{loop-cois} and \refToThm{cycle-is}. 
\end{proof}

The resulting algorithm behaves as follows: 
we start from a judgement $\judg$ with an empty set of circular \EZ{hypotheses}, 
then we try to build a regular derivation for $\judg$ using rules in $\is$, exactly the same way as for standard regular coinduction; 
but, this time, when we find a cycle, say for a judgement $\judg'$,  we trigger another procedure, which looks for a finite derivation  for $\judg'$ in $\Ind{\is\cup\cois}$. 

\FDComm{vedere se dire che nella semantica proof-theoretic basta un albero finito per un nodo preso da ogni ciclo} 

\subsection{Flexible regular reasoning}
We now adapt proof techniques presented in \refToSect{reasoning} to this generalised setting. 
For completeness proofs, in \cite{AnconaDZ17esop,Dagnino19}, the standard coinduction principle is extended to generalised inference systems, by adding an additional constraint, which takes into account corules. 
The regular coinduction principle (\refToProp{reg-coind}) can be smoothly extended to this generalised setting following the same strategy, as expressed in the next proposition. 
We call the resulting principle the \emph{bounded regular coinduction principle}.

\begin{prop}[Bounded regular coinduction] \label{prop:breg-coind} 
Let $\Spec \subseteq \universe$ be a set of judgements, then 
if, for all $\judg \in \Spec$, there is a finite set $X \subseteq \universe$ such that 
\begin{itemize}
\item $X\subseteq\Ind{\is\cup\cois}$, 
\item $X \subseteq \InfOp{\is}(X)$, and 
\item $\judg \in X$, 
\end{itemize} 
then, $\Spec \subseteq \FlexReg{\is}{\cois}$. 
\end{prop}
This proposition immediately follows from  \refToProp{below-rfp}, as 
$\FlexReg{\is}{\cois}$ is a rational fixed point by \refToProp{brfp-corules} and \refToDef{brfp}. 
The additional constraint $\Spec\subseteq\Ind{\is\cup\cois}$, named \emph{boundedness}, reflects the fact that, using corules, we are only allowed to build proof trees using judgements in $\Ind{\is\cup\cois}$. 
Note that, when $\cois = \is_\universe$, thus $\FlexReg{\is}{\cois} = \Reg{\is}$, the additional constraint is trivially true, because it requires $\Spec\subseteq \universe$, 
hence we recover the regular coinduction principle in \refToProp{reg-coind}. 

We illustrate this technique on our running example: the definition of $\minElem{x}{s}$, which should hold when $x$ is the minimum of the rational stream of integers $s$. 
We denote by $\minSpec$ the set of judgements $\minElem{x}{s}$ where $x$ is indeed the minimum of $s$. 
We prove, using \refToProp{breg-coind}, the following statement:
\begin{quote}
if $\minElem{x}{s}\in \minSpec$ then it has a regular derivation with corules
\end{quote}
\begin{proof}
Let $\minElem{x}{s}\in\minSpec$ and define $X$ as the set of judgements $\minElem{z}{r}\in\minSpec$ such that $r$ is a substream of $s$. 
Trivially $\minElem{x}{s}\in X$ and, since $s$ is rational, it has finitely many different substreams, hence $X$ is finite. 
The boundedness condition, that is, if $\minElem{z}{r}\in X$ then it has a finite proof tree using also the coaxioms,  is easy to check, because, 
if $y$ is the minimum of $r$, then $y$ occurs somewhere in $r$,
hence we can prove the thesis  by induction on the least position of $y$ in $r$.
In order to check that $X$ is consistent, consider $\minElem{z}{r}\in X$, with $r = y\cons r'$
Since $z$ is the minimum of $r$ and $r'$ is a substream of $r$, $z$ is a lower bound of $r'$, thus it has a minimum, say $y'$, and so $\minElem{y'}{r'}\in X$. 
To conclude, we have to show that $z = \min\{y,y'\}$. 
The inequality $z\le \min\{y,y'\}$ is trivial, 
for the other inequality, since $z$ belongs to $r$, we have two cases:
if $z = y$, then $\min\{y,y'\}\le z$, 
otherwise $z$ belongs to $r'$ and so $y'\le z$, thus $\min\{y,y'\}\le z$. 
\end{proof}

Differently from the standard coinductive interpretation, for the regular interpretation we have also defined a proof technique to show soundness (\refToProp{sound-ind}). 
Such a technique relies on the inductive characterisation of the regular interpretation. 
As also the regular interpretation of an inference system with corules has an inductive characterisation (\refToCor{loop-cois}), we can provide a proof technique to show soundness also in this generalised setting, which smoothly extends the one of standard regular coinduction. 

\begin{prop} \label{prop:sound-ind-corules}
Let $\Spec \subseteq \universe$ be a set of judgements, then, 
if there is a family \EZ{$(\Spec_\LoopHp)_{\LoopHp \in \finwp(\universe)}$} such that $\Spec_\LoopHp \subseteq \universe$ and $\Spec_\emptyset \subseteq  \Spec$, and, 
for all $\LoopHp \in \finwp(\universe)$, 
\begin{itemize}
\item $\LoopHp\cap\Ind{\is\cup\cois} \subseteq \Spec_\LoopHp$, and 
\item for all rules $\RulePair{\prem}{\judg} \in \is$, if $\prem \subseteq \Spec_{\LoopHp \cup\{\judg\}}$ then $\judg \in \Spec_\LoopHp$, 
\end{itemize}
then $\FlexReg{\is}{\cois}  \subseteq \Spec$. 
\end{prop}
\begin{proof}
By a straightforward induction on rules in $\Loopcois{\is}{\cois}$, we get that if $\validInd{\Loopcois{\is}{\cois}}{\LoopJ{\LoopHp}{\judg}}$, then $\judg\in\Spec_\LoopHp$; 
thus, the thesis follows from \refToCor{loop-cois}. 
\end{proof}
Again, this proof principle is almost the same as \refToProp{sound-ind}, but with an additional constraint, this time on sets  of circular hypothesis, which takes into account corules. 

We illustrate this technique proving that the definition of $\minElem{x}{s}$ is sound, that is, 
\begin{quote}
if $\minElem{x}{s}$ has a regular derivation with corules then $x$ is the minimum of $s$. 
\end{quote}
\begin{proof} 
First of all, we note that, if $\minElem{x}{s}$ has a finite proof tree using also the coaxiom, then $x$ belongs to $s$. 
Then, we define $\minSpec_\LoopHp$ as follows: 
$\minElem{x}{s}\in \minSpec_\LoopHp$ iff $x$ is the minimum of $s$ or $s = x_1\cons\ldots\cons x_n\cons r$, $\minElem{y}{r}\in\LoopHp$, $\minElem{y}{r}$ has a finite proof tree using also the coaxiom  and $x = \min\{x_1,\ldots,x_n,y\}$. 
We have trivialy that $\minSpec_\emptyset\subseteq \minSpec$. 

Assume a finite set of judgements $\LoopHp$. 
Clearly, if $\minElem{x}{s}\in \LoopHp$ has a finite proof tree using also the coaxiom, then $\minElem{x}{s}\in\minSpec_\LoopHp$. 
Now, suppose $s = x\cons r$, $\LoopHp'=\LoopHp\cup\{\minElem{z}{s}\}$, $\minElem{y}{r}\in\minSpec_{\LoopHp'}$ and $z = \min\{x,y\}$, then we have two cases:
\begin{itemize}
\item if $y$ is the minimum of $r$, then $z$ is the minimum of $s = y\cons r$, hence $\minElem{z}{s}\in\minSpec_\LoopHp$; 
\item if $r = x_1\cons\ldots\cons x_n\cons r'$, $\minElem{y'}{r'}\in\LoopHp'$, $\minElem{y'}{r'}$ has a finite proof tree using also the coaxiom  and $y = \min\{x_1,\ldots,x_n,y'\}$, then  
$s = x\cons r = x\cons x_1\cons \ldots\cons x_n\cons r'$ and $z = \min\{x,x_1,\ldots,x_n,y'\}$. 
We distinguish two subcases:
\begin{itemize}
\item if $\minElem{y'}{r'}\in\LoopHp$, then $\minElem{z}{s}\in\minSpec_\LoopHp$ by definition, and 
\item if $y' = z$ and $r' = s$, then $s = x\cons x_1\cons\ldots\cons x_n\cons s$ and $\minElem{z}{s}$ has a finite proof tree using also the coaxiom, thus $z$ belongs to $s$ and  $z = \min\{x,x_1,\ldots,x_n,z\}$, that is, $z$ is the minimum of $s$. 
  \qedhere
\end{itemize}
\end{itemize}
\end{proof}

We now consider a more involved example. 
\begin{exa}[Addition of rational numbers]
It is well-known that real numbers in $[0,1]$  can be represented as, not necessarily rational,  streams of digits in some basis. 
Let $\NPos$ be the set of positive natural numbers and assume a basis $b\in\NPos$. 
A digit $d$ is a natural number in $0..b-1$, then, 
given a stream $r = (d_i)_{i\in\NPos}$ of digits, the series $\sum_{i=1}^\infty d_ib^{-i}$ converges and its limit is the real number represented by $r$ and denoted by $\sem{r}$. 
It is also well-known that every real number $x\in[0,1]$ has at most two different representations as a stream, 
for instance, with $b = 10$, the number $1/2$ can be represented as either $5\cons \rep 0$ or $4\cons \rep 9$, where, for any digit $d$, $\rep d$ is the stream $d\cons d \cons d \cons \ldots$. 

Consider the following inference system with corules, defining the judgement $\addNum{r_1}{r_2}{r}{c}$, where $c$ is an integer representing the carry, and which should hold when $\sem{r_1}+\sem{r_2} = \sem{r}+c$. 
\[
\Rule{
  \addNum{r_1}{r_2}{r}{c}
}{ \addNum{d_1\cons r_1}{d_2\cons r_2}{(x\mod b)\cons r}{x\div b} }\, x = d_1+d_2+c
\BigSpace
\CoAxiom{ \addNum{r_1}{r_2}{r}{c} }\, c\in -1..2
\]
In \cite{AnconaDZ17esop,Dagnino19}  it is proved that this definition is correct. 
It is also well-known that rational streams of digits represent rational numbers, that is, if $r$ is a rational stream of digits, then $\sem{r}$ is a  rational number. 
We show here that the regular interpretation of the above inference system with corules is correct with respect to the addition of rational numbers. 

Define the set $\addSpec$ of correct judgements as follows:
$\addNum{r_1}{r_2}{r}{c} \in\addSpec$ iff $r_1$, $r_2$ and $r$ are rational and $\sem{r_1}+\sem{r_2} = \sem{r}+c$. 
We start by proving completeness, stated below:
\begin{quote}
for rational streams $r_1,r_2,r$, if $\sem{r_1}+\sem{r_2} = \sem{r}+c$ then $\addNum{r_1}{r_2}{r}{c}$ has a regular derivation with corules. 
\end{quote}
We use the bounded regular coinduction principle.
First of all, note that, since $\sem{r}\in [0,1]$ for any stream $r$, if $\addNum{r_1}{r_2}{r}{c}\in\addSpec$, then 
$c = \sem{r_1}+\sem{r_2}-\sem{r}$, hence $c\ge -1$ and $c\le 2$. 
Therefore, we immediately have that all judgements in $\addSpec$ have a finite proof tree using also the coaxiom. 

Assume $\addNum{r_1}{r_2}{r}{c}\in\addSpec$ and define \EZ{$X$} as follows:
$\addNum{s_1}{s_2}{s}{c'}\in X$ iff $\sem{s_1}+\sem{s_2}=\sem{s}+c'$ and $s_1$ and $s_2$ are substreams of $r_1$ and $r_2$, respectively. 
Trivially, $\addNum{r_1}{r_2}{r}{c}\in X$ and $X$ is finite because, 
since $r_1$ and $r_2$ are rational, they have finitely many different substreams, and $c'$ can assume only four values, hence $\sem{s} = \sem{s_1}+\sem{s_2}-c'$ can assume only finitely many values, and so there are finitely many $s$ satisfying that equation. 
Now we have to check that $X$ is  consistent.
Assume $\addNum{d_1\cons s_1}{d_2\cons s_2}{d\cons s}{c'} \in X$, then we have 
$\sem{d_1\cons s_1}+\sem{d_2\cons s_2} = \sem{d\cons s} +c'$. 
It is easy to check that, for any stream $t$ and digit $d$, $b\sem{d\cons t} = d + \sem{t}$. 
Hence, we get $\sem{s_1}+\sem{s_2} = \sem{s} + c''$, with $c'' = bc' + d -d_1-d_2$. 
Since $s_1$ and $s_2$ are still substream of $r_1$ and $r_2$, respectively, we get $\addNum{s_1}{s_2}{s}{c''} \in X$, as needed. 

We now prove soundness, as stated below:
\begin{quote}
if $\addNum{r_1}{r_2}{r}{c}$ has a regular derivation with corules, then $\sem{r_1}+\sem{r_2} = \sem{r}+c$ and $r_1$,$r_2$ and $r$ are rational stream. 
\end{quote}
For any finite set $\LoopHp$, we define $\addSpec_\LoopHp$ as follows: 
$\addNum{r_1}{r_2}{r}{c_0}\in \addSpec_\LoopHp$ iff $r_1 = d_{11}\cons\ldots\cons d_{1n}\cons s_1$, 
$r_2 = d_{21}\cons\ldots\cons d_{2n}\cons s_2$, 
$r = d_1\cons\ldots\cons d_n\cons s$ and,
there are $c_1,\ldots,c_n \in -1..2$ such that, 
for all $i\in 1..n$, $d_{1i}+d_{2i}+c_i = bc_{i-1} + d_i$, and 
either $\addNum{s_1}{s_2}{s}{c_n} \in \LoopHp$ or $s_1 = r_1$, $s_2 = r_2$, $s = r$ and $c_0 = c_n$. 
The two closure properties in \refToProp{sound-ind-corules} are easy to check. 
Hence, to conclude it is enough we show that $\addSpec_\emptyset \subseteq \addSpec$. 
To this end, assume $\addNum{r_1}{r_2}{r}{c_0}\in \addSpec_\emptyset$, then, by definition, we have 
$r_1 = d_{11}\cons\ldots \cons d_{1n}\cons r_1$, 
$r_2 = d_{21}\cons \ldots \cons d_{2n} \cons r_2$, 
$r   = d_1 \cons \ldots \cons d_n \cons r$ and, 
there are $c_1,\ldots,c_n \in -1..2$ such that, 
for all $i \in 1..n$,  $d_{1i}+d_{2i}+c_i = bc_{i-1}+d_i$. 
This implies that 
$r_1=(d_{1_i})_{i\in\NPos}$, $r_2=(d_{2i})_{i\in\NPos}$ and $r=(d_i)_{i\in\NPos}$ are rational streams and, 
for all $i\in\NPos$, $d_{1i}+d_{2i} + c_{j+1} = bc_j +d_i$ where $j = i\mod n$. 
Hence, we have only to check that 
$\sem{r_1}+\sem{r_2} = \sem{r}+c_0$. 
We define sequences $(x_k)_{k\in\NPos}$ and $(y_k)_{k\in\NPos}$ as 
$x_k = \sum_{i=1}^k (d_{1i}+d_{2i})b^{-i}$ and $y_k = \sum_{i=1}^k d_ib^{-i}$. 
Then we have to show that $\lim x_k -\lim y_k = \lim (x_k-y_k) = c_0$, 
because $\sem{r_1}+\sem{r_2} = \lim x_k$ and $\sem{r} = \lim y_k$. 
As, for all $k\in\NPos$ we have $d_{1k}+d_{2k}+c_{j+1} = bc_j +d_k$ with $j = k \mod n$, we get 
$c_0 - (x_k - y_k) = c_0 + y_k - x_k = c_{j+1}b^{-k}$, that, when $k$ tends to $\infty$, converges to $0$, hence  we get 
$\lim (x_k-y_k) = c_0$. 
\end{exa}

\section{Related work} \label{sect:related}

The regular approach has been adopted in many different contexts, notably to define proof systems for several kinds of logics, and to define operational models of programming languages supporting cyclic structures. 

Concerning proof systems allowing regular proofs, usually called \emph{circlar proofs},  
we find proposals in  \cite{Santocanale02,FortierS13,Doumane17} for logics with fixed point operators, and in \cite{Brotherston05,BrotherstonS11} for classical first order logic with inductive definitions. 
In both cases, regular proofs allow to naturally handle the unfolding of fixed point and recursive definitions, respectively. 
However, regular proofs allow the derivation of wrong sequents, such as the empty one; 
hence, to solve this issue, they have to impose additional constraints on regular proofs, such as parity conditions in \cite{FortierS13}, thus filtering out undesired derivations. 
These additional requirements on regular proofs  are expressed at the meta-level and typically require some condition to hold infinitely often in the regular proof.
As inference systems with corules have been designed precisely to filter out undesired infinite derivations, 
and they seem pretty good at capturing requirements that should hold infinitely often in the proof, 
it would be interesting to investigate whether these additional constraints can be enforced by an appropriate set of corules.

In \cite{BasoldKL19}, the authors present a proof system supporting coinductive reasoning with Horn clauses. 
Such a system allows only finite derivations, enabling coinduction thanks to a cycle detection mechanism similar to the one we use in the context of inference systems. 
However, in their setting, they have to perform an additional `productivity'' check, that is, they have to ensure that a cycle  is closed only after the application of a clause, otherwise they would fall in  inconsistency. 
This check is not required in our setting, as our derivations are built only using rules. 

Even though all these proof systems are tightly related to  our work, there is an important difference: 
we study regular reasoning by itself,  without fixing a specific syntax. 
In this way, we can work abstractly, focusing only on the essential feature of regularity, thus providing a common abstract and simple background  to all these proof systems. 

The other context where we can find applications of regular coinduction is in programming languages supporting cyclic structures. 
In this case, we use the term \emph{regular corecursion} for a  semantics of recursive definitions which detects cycles, analogously to the inductive characterization of the regular interpretation in \refToSect{cycle}. 

We can find proposals of language constructs for regular corecursion in all common programming paradigms: logic \cite{SimonMBG06,SimonBMG07,AnconaD15},  functional  \cite{JeanninKS13,JeanninKS17} and object-oriented  \cite{AnconaZ12}. 
There are also proposals, inspired by corules, supporting a flexible form of regular corecursion in the logic \cite{AnconaDZ17coalp,DagninoAZ20} and object-oriented \cite{AnconaBDZ20} paradigms.

The approach proposed in the logic paradigm is particularly interesting. 
Indeed, logic programs can be regarded as particular, syntactic instances of inference systems: \EZ{judgments are ground atoms, and inference rules are ground instances of clauses.}

The declarative semantics of logic programs \EZ{turns out to be defined exactly} the same way as \EZ{that of the underlying} inference systems, as fixed point of the associated inference operator. 
Moreover, 
in coinductive logic programming, the resolution procedure, called \emph{coSLD resolution} \cite{SimonMBG06,SimonBMG07,AnconaD15}, keeps track of already encountered goals, so that, 
if it finds again the same goal, up-to unification, it can accept it. 
This mechanism looks very similar to our inductive characterization of the regular interpretation, hence, the analogy with inference systems holds also at the operational level. The same analogies hold between generalised inference systems and \emph{logic programs with coclauses} introduced in \cite{DagninoAZ20}. 

\EZ{Basing on this correspondence, in \cite{DagninoAZ20} soundness and completeness of the resolution procedure with respect to the regular declarative semantics are proved, rather than in an ad-hoc way, by relying on the inductive characterisation of the regular interpretation given in this paper (\refToThm{cycle-is} and \refToCor{loop-cois}). That is, it is enough to show that the resolution procedure is equivalent to the inductive characterisation of the regular interpretation.}\EZComm{dire in future work che cercheremo di farlo in altri casi? tipo quello funzionale?}

\section{Conclusion} \label{sect:conclu} 

Inference systems \EZ{\cite{Aczel77,Sangiorgi11}} are \EZ{a} widely used framework to define and reason on several kinds of judgements (small-step and big-step operational semantics, type systems, proof systems, etc.) by means of inference rules. 
They naturally support both inductive and coinductive reasoning: in the former approach only finite derivations are allowed, while in the latter one, arbitrary derivations (finite or not) can be used. 

In this paper, we have considered regular reasoning, an interesting middle way between induction and coinduction, combining advantages of both approaches: 
it is not restricted to finite derivations, thus overcoming limitations of the inductive approach,  but it still has a finite nature, as a regular derivation can only contain finitely many judgements. 
We started from a natural proof-theoretic definition of the regular interpretation of an inference system, as the set of judgements derivable by a regular proof tree. 
After presenting the construction of the rational fixed point in a lattice-theoretic setting, we proved that the regular interpretation coincides with the rational fixed point of the inference operator associated with the inference system. 
Then, we showed that the regular interpretation has an equivalent inductive characterization, which provides us with an algorithm to \EZ{find a derivation for a judgment, if any}. 
Relying on these results, we discussed proof techniques for regular reasoning: from the fact that the regular interpretation is a rational fixed point, we got the regular coinduction principle, which allows us to prove completeness, 
while, from the inductive characterization, we derived a proof technique to show soundness. 

Finally, we focused on inference systems with corules \cite{AnconaDZ17esop,Dagnino19}, a recently introduced generalisation of inference systems, allowing refinements of the coinductive interpretation. 
We \EZ{showed} that all results presented for regular coinduction can be smoothly extended to this generalised framework, thus providing a flexible approach also to regular reasoning. 

Concerning future work, 
an interesting direction is the development of more sophisticated proof techniques for regular reasoning. 
Indeed, several enhanced coinductive techniques have been proposed, such as \emph{parametrized coinduction} \cite{HurNDV13} and \emph{coinduction up-to} \cite{BonchiPPR17}, which have been proved to be effective in several contexts. 
Adapting such techniques to the (flexible) regular case would provide us with powerful tools to support regular reasoning. 
A further development in this direction would be to provide support to regular reasoning in proof assistants, which usually provide primitives only for plain induction and coinduction. 
To this end, we could start from existing approaches \cite{Spadotti16,UustaluV17} to implement regular terms in proof assistants.
Finally, it would be interesting to apply results in this paper to build, in a principled way,  abstract and operational models of languages supporting regular corecursion, going beyond the logic paradigm.

\bibliographystyle{alphaurl}
\bibliography{bib}

\begin{thebibliography}{AAMV03}

\bibitem[AAMV03]{AczelAMV03}
Peter Aczel, Jir{\'{\i}} Ad{\'{a}}mek, Stefan Milius, and Jiri Velebil.
\newblock Infinite trees and completely iterative theories: a coalgebraic view.
\newblock {\em Theoretical Computer Science}, 300(1-3):1--45, 2003.
\newblock \href {https://doi.org/10.1016/S0304-3975(02)00728-4}
  {\path{doi:10.1016/S0304-3975(02)00728-4}}.

\bibitem[ABDZ20]{AnconaBDZ20}
Davide Ancona, Pietro Barbieri, Francesco Dagnino, and Elena Zucca.
\newblock Sound regular corecursion in {coFJ}.
\newblock In {\em 34nd European Conference on Object-Oriented Programming,
  {ECOOP} 2020}, volume 166 of {\em LIPIcs}. Schloss Dagstuhl - Leibniz-Zentrum
  fuer Informatik, 2020.
\newblock \href {https://doi.org/10.4230/LIPIcs.ECOOP.2020.1}
  {\path{doi:10.4230/LIPIcs.ECOOP.2020.1}}.

\bibitem[Acz77]{Aczel77}
Peter Aczel.
\newblock An introduction to inductive definitions.
\newblock In Jon Barwise, editor, {\em Handbook of Mathematical Logic},
  volume~90 of {\em Studies in Logic and the Foundations of Mathematics}, pages
  739 -- 782. Elsevier, 1977.

\bibitem[AD15]{AnconaD15}
Davide Ancona and Agostino Dovier.
\newblock A theoretical perspective of coinductive logic programming.
\newblock {\em Fundamenta Informaticae}, 140(3-4):221--246, 2015.
\newblock \href {https://doi.org/10.3233/FI-2015-1252}
  {\path{doi:10.3233/FI-2015-1252}}.

\bibitem[ADZ17a]{AnconaDZ17coalp}
Davide Ancona, Francesco Dagnino, and Elena Zucca.
\newblock Extending coinductive logic programming with co-facts.
\newblock In Ekaterina Komendantskaya and John Power, editors, {\em Proceedings
  of the First Workshop on Coalgebra, Horn Clause Logic Programming and Types,
  {CoALP-Ty} 2016}, volume 258 of {\em Electronic Proceedings in Theoretical
  Computer Science}, pages 1--18. Open Publishing Association, 2017.
\newblock \href {https://doi.org/10.4204/EPTCS.258.1}
  {\path{doi:10.4204/EPTCS.258.1}}.

\bibitem[ADZ17b]{AnconaDZ17esop}
Davide Ancona, Francesco Dagnino, and Elena Zucca.
\newblock Generalizing inference systems by coaxioms.
\newblock In Hongseok Yang, editor, {\em Programming Languages and Systems -
  26th European Symposium on Programming, {ESOP} 2017}, volume 10201 of {\em
  Lecture Notes in Computer Science}, pages 29--55. Springer, 2017.
\newblock \href {https://doi.org/10.1007/978-3-662-54434-1_2}
  {\path{doi:10.1007/978-3-662-54434-1_2}}.

\bibitem[AMV06]{AdamekMV06}
Jir{\'{\i}} Ad{\'{a}}mek, Stefan Milius, and Jiri Velebil.
\newblock Iterative algebras at work.
\newblock {\em Mathematical Structures in Computer Scienc}, 16(6):1085--1131,
  2006.
\newblock \href {https://doi.org/10.1017/S0960129506005706}
  {\path{doi:10.1017/S0960129506005706}}.

\bibitem[AZ12]{AnconaZ12}
Davide Ancona and Elena Zucca.
\newblock Corecursive {F}eatherweight {J}ava.
\newblock In Wei{-}Ngan Chin and Aquinas Hobor, editors, {\em Proceedings of
  the 14th Workshop on Formal Techniques for Java-like Programs, {FTfJP} 2012},
  pages 3--10. {ACM} Press, 2012.
\newblock \href {https://doi.org/10.1145/2318202} {\path{doi:10.1145/2318202}}.

\bibitem[BKL19]{BasoldKL19}
Henning Basold, Ekaterina Komendantskaya, and Yue Li.
\newblock Coinduction in uniform: Foundations for corecursive proof search with
  horn clauses.
\newblock In Lu{\'{\i}}s Caires, editor, {\em Programming Languages and Systems
  - 28th European Symposium on Programming, {ESOP} 2019}, volume 11423 of {\em
  Lecture Notes in Computer Science}, pages 783--813. Springer, 2019.
\newblock \href {https://doi.org/10.1007/978-3-030-17184-1\_28}
  {\path{doi:10.1007/978-3-030-17184-1\_28}}.

\bibitem[BPPR17]{BonchiPPR17}
Filippo Bonchi, Daniela Petrisan, Damien Pous, and Jurriaan Rot.
\newblock A general account of coinduction up-to.
\newblock {\em Acta Informatica}, 54(2):127--190, 2017.
\newblock \href {https://doi.org/10.1007/s00236-016-0271-4}
  {\path{doi:10.1007/s00236-016-0271-4}}.

\bibitem[Bro05]{Brotherston05}
James Brotherston.
\newblock Cyclic proofs for first-order logic with inductive definitions.
\newblock In Bernhard Beckert, editor, {\em Automated Reasoning with Analytic
  Tableaux and Related Methods, International Conference, {TABLEAUX} 2005},
  volume 3702 of {\em Lecture Notes in Computer Science}, pages 78--92.
  Springer, 2005.
\newblock \href {https://doi.org/10.1007/11554554\_8}
  {\path{doi:10.1007/11554554\_8}}.

\bibitem[BS11]{BrotherstonS11}
James Brotherston and Alex Simpson.
\newblock Sequent calculi for induction and infinite descent.
\newblock {\em Journal of Logic and Computation}, 21(6):1177--1216, 2011.
\newblock \href {https://doi.org/10.1093/logcom/exq052}
  {\path{doi:10.1093/logcom/exq052}}.

\bibitem[CC77]{CousotC77}
Patrick Cousot and Radhia Cousot.
\newblock Abstract interpretation: {A} unified lattice model for static
  analysis of programs by construction or approximation of fixpoints.
\newblock In Robert~M. Graham, Michael~A. Harrison, and Ravi Sethi, editors,
  {\em The 4th {ACM} Symposium on Principles of Programming Languages, {POPL}
  '77}, pages 238--252. {ACM}, 1977.
\newblock \href {https://doi.org/10.1145/512950.512973}
  {\path{doi:10.1145/512950.512973}}.

\bibitem[Cou83]{Courcelle83}
Bruno Courcelle.
\newblock Fundamental properties of infinite trees.
\newblock {\em Theoretical Computer Science}, 25:95--169, 1983.
\newblock \href {https://doi.org/10.1016/0304-3975(83)90059-2}
  {\path{doi:10.1016/0304-3975(83)90059-2}}.

\bibitem[Dag17]{Dagnino17}
Francesco Dagnino.
\newblock Generalizing inference systems by coaxioms.
\newblock Master's thesis, DIBRIS, University of Genova, 2017.
\newblock Best italian master thesis in Theoretical Computer Science 2018.
\newblock URL: \url{http://eatcs.org/images/it/theses/Dagnino_thesis.pdf}.

\bibitem[Dag19]{Dagnino19}
Francesco Dagnino.
\newblock Coaxioms: flexible coinductive definitions by inference systems.
\newblock {\em Logical Methods in Computer Science}, 15(1), 2019.
\newblock \href {https://doi.org/10.23638/LMCS-15(1:26)2019}
  {\path{doi:10.23638/LMCS-15(1:26)2019}}.

\bibitem[DAZ20]{DagninoAZ20}
Francesco Dagnino, Davide Ancona, and Elena Zucca.
\newblock Flexible coinductive logic programming.
\newblock {\em Theory and Practice of Logic Programming}, 20(6):818--833, 2020.
\newblock Issue for ICLP 2020.
\newblock \href {https://doi.org/10.1017/S147106842000023X}
  {\path{doi:10.1017/S147106842000023X}}.

\bibitem[Dou17]{Doumane17}
Amina Doumane.
\newblock {\em On the infinitary proof theory of logics with fixed points.
  (Th{\'{e}}orie de la d{\'{e}}monstration infinitaire pour les logiques
  {\`{a}} points fixes)}.
\newblock PhD thesis, Paris Diderot University, France, 2017.
\newblock URL: \url{https://tel.archives-ouvertes.fr/tel-01676953}.

\bibitem[DP02]{DaveyP02}
Brian~A. Davey and Hilary~A. Priestley.
\newblock {\em Introduction to Lattices and Order}.
\newblock Cambridge University Press, 2 edition, 2002.
\newblock \href {https://doi.org/10.1017/CBO9780511809088}
  {\path{doi:10.1017/CBO9780511809088}}.

\bibitem[FS13]{FortierS13}
J{\'{e}}r{\^{o}}me Fortier and Luigi Santocanale.
\newblock Cuts for circular proofs: semantics and cut-elimination.
\newblock In Simona Ronchi~Della Rocca, editor, {\em Computer Science Logic
  2013, {CSL} 2013}, volume~23 of {\em LIPIcs}, pages 248--262. Schloss
  Dagstuhl - Leibniz-Zentrum fuer Informatik, 2013.
\newblock \href {https://doi.org/10.4230/LIPIcs.CSL.2013.248}
  {\path{doi:10.4230/LIPIcs.CSL.2013.248}}.

\bibitem[HNDV13]{HurNDV13}
Chung{-}Kil Hur, Georg Neis, Derek Dreyer, and Viktor Vafeiadis.
\newblock The power of parameterization in coinductive proof.
\newblock In Roberto Giacobazzi and Radhia Cousot, editors, {\em The 40th
  Annual {ACM} Symposium on Principles of Programming Languages, {POPL} '13},
  pages 193--206. {ACM} Press, 2013.
\newblock \href {https://doi.org/10.1145/2429069.2429093}
  {\path{doi:10.1145/2429069.2429093}}.

\bibitem[JKS13]{JeanninKS13}
Jean{-}Baptiste Jeannin, Dexter Kozen, and Alexandra Silva.
\newblock Language constructs for non-well-founded computation.
\newblock In Matthias Felleisen and Philippa Gardner, editors, {\em Programming
  Languages and Systems - 22nd European Symposium on Programming, {ESOP} 2013},
  volume 7792 of {\em Lecture Notes in Computer Science}, pages 61--80.
  Springer, 2013.
\newblock \href {https://doi.org/10.1007/978-3-642-37036-6_4}
  {\path{doi:10.1007/978-3-642-37036-6_4}}.

\bibitem[JKS17]{JeanninKS17}
Jean{-}Baptiste Jeannin, Dexter Kozen, and Alexandra Silva.
\newblock Cocaml: Functional programming with regular coinductive types.
\newblock {\em Fundamenta Informaticae}, 150(3-4):347--377, 2017.
\newblock \href {https://doi.org/10.3233/FI-2017-1473}
  {\path{doi:10.3233/FI-2017-1473}}.

\bibitem[LG09]{LeroyG09}
Xavier Leroy and Herv\'e Grall.
\newblock Coinductive big-step operational semantics.
\newblock {\em Information and Computation}, 207(2):284--304, 2009.
\newblock \href {https://doi.org/10.1016/j.ic.2007.12.004}
  {\path{doi:10.1016/j.ic.2007.12.004}}.

\bibitem[MPW16]{MiliusPW16}
Stefan Milius, Dirk Pattinson, and Thorsten Wi{\ss}mann.
\newblock A new foundation for finitary corecursion - the locally finite
  fixpoint and its properties.
\newblock In Bart Jacobs and Christof L{\"{o}}ding, editors, {\em Foundations
  of Software Science and Computation Structures - 19th International
  Conference, {FOSSACS} 2016}, volume 9634 of {\em Lecture Notes in Computer
  Science}, pages 107--125. Springer, 2016.
\newblock \href {https://doi.org/10.1007/978-3-662-49630-5\_7}
  {\path{doi:10.1007/978-3-662-49630-5\_7}}.

\bibitem[MPW19]{MiliusPW19}
Stefan Milius, Dirk Pattinson, and Thorsten Wi{\ss}mann.
\newblock A new foundation for finitary corecursion and iterative algebras.
\newblock {\em Information and Computation}, page 104456, 2019.
\newblock \href {https://doi.org/10.1016/j.ic.2019.104456}
  {\path{doi:10.1016/j.ic.2019.104456}}.

\bibitem[San02]{Santocanale02}
Luigi Santocanale.
\newblock A calculus of circular proofs and its categorical semantics.
\newblock In Mogens Nielsen and Uffe Engberg, editors, {\em Foundations of
  Software Science and Computation Structures - 5th International Conference,
  {FOSSACS} 2002}, volume 2303 of {\em Lecture Notes in Computer Science},
  pages 357--371. Springer, 2002.
\newblock \href {https://doi.org/10.1007/3-540-45931-6\_25}
  {\path{doi:10.1007/3-540-45931-6\_25}}.

\bibitem[San11]{Sangiorgi11}
Davide Sangiorgi.
\newblock {\em Introduction to Bisimulation and Coinduction}.
\newblock Cambridge University Press, USA, 2011.

\bibitem[SBMG07]{SimonBMG07}
Luke Simon, Ajay Bansal, Ajay Mallya, and Gopal Gupta.
\newblock Co-logic programming: Extending logic programming with coinduction.
\newblock In Lars Arge, Christian Cachin, Tomasz Jurdzinski, and Andrzej
  Tarlecki, editors, {\em Automata, Languages and Programming, 34th
  International Colloquium, {ICALP} 2007}, volume 4596 of {\em Lecture Notes in
  Computer Science}, pages 472--483. Springer, 2007.
\newblock \href {https://doi.org/10.1007/978-3-540-73420-8_42}
  {\path{doi:10.1007/978-3-540-73420-8_42}}.

\bibitem[SMBG06]{SimonMBG06}
Luke Simon, Ajay Mallya, Ajay Bansal, and Gopal Gupta.
\newblock Coinductive logic programming.
\newblock In Sandro Etalle and Miroslaw Truszczynski, editors, {\em Logic
  Programming, 22nd International Conference, {ICLP} 2006}, volume 4079 of {\em
  Lecture Notes in Computer Science}, pages 330--345. Springer, 2006.
\newblock \href {https://doi.org/10.1007/11799573_25}
  {\path{doi:10.1007/11799573_25}}.

\bibitem[Spa16]{Spadotti16}
R{\'{e}}gis Spadotti.
\newblock {\em A mechanized theory of regular trees in dependent type theory.
  (Une th{\'{e}}orie m{\'{e}}canis{\'{e}}e des arbres r{\'{e}}guliers en
  th{\'{e}}orie des types d{\'{e}}pendants)}.
\newblock PhD thesis, Paul Sabatier University, Toulouse, France, 2016.
\newblock URL: \url{https://tel.archives-ouvertes.fr/tel-01589656}.

\bibitem[Tar55]{Tarski55}
Alfred Tarski.
\newblock A lattice-theoretical fixpoint theorem and its applications.
\newblock {\em Pacific Journal of Mathematics}, 5(2):285--309, 1955.

\bibitem[UV17]{UustaluV17}
Tarmo Uustalu and Niccol{\`{o}} Veltri.
\newblock Finiteness and rational sequences, constructively.
\newblock {\em Journal of Functional Programming}, 27:e13, 2017.
\newblock \href {https://doi.org/10.1017/S0956796817000041}
  {\path{doi:10.1017/S0956796817000041}}.

\end{thebibliography}

\end{document}